\documentclass[
aps, 
pra, 
showpacs,
nofootinbib, 
superscriptaddress, 
10pt, 
twocolumn]{revtex4-1}

\usepackage[utf8]{inputenc}
\usepackage{amsmath}
\usepackage{stmaryrd}
\usepackage{color}
\usepackage{mathrsfs}
\usepackage{yfonts}
\usepackage{rotating}
\usepackage{graphicx}
\usepackage{amsfonts}
\usepackage{mathrsfs}
\usepackage[american]{babel}
\usepackage{wasysym}
\usepackage{slashed}
\usepackage{pgfplots}
\usepackage{braket}
\usepackage{graphicx}
\usepackage{dsfont}
\usepackage{amsthm}
\usepackage{tabu}
\usepackage{subcaption}
\usepackage{caption}
\usepackage{ulem}
\usepackage{nicematrix}
\usepackage{booktabs} 
\usepackage{physics}

\usepackage{makecell}

\usepackage{natbib} 
\PassOptionsToPackage{breaklinks}{hyperref}
\usepackage[colorlinks,allcolors=black]{hyperref}
\usepackage{orcidlink}

\usepackage[capitalise,noabbrev]{cleveref}

\theoremstyle{definition}
\newtheorem{definition}{Definition}
\newtheorem{example}[definition]{Example}

\newtheorem{lemma}[definition]{Lemma}

\newtheorem{proposition}[definition]{Proposition}
\newtheorem{theorem}[definition]{Theorem}
\newtheorem{corollary}[definition]{Corollary}
\newtheorem{conjecture}[definition]{Conjecture}

 \crefname{equation}{Eq.}{Eqs.}

\newcommand\Odd[1]{\text{Odd}\left( #1  \right) }
\DeclareMathOperator{\supp}{supp}
\DeclareMathOperator{\id}{Id}
\DeclareMathOperator{\LC}{LC}
\DeclareMathOperator{\LU}{LU}
\DeclareMathOperator{\Id}{Id}
\DeclareMathOperator{\rk}{rank}
\DeclareMathOperator{\CutRank}{CutRank}
\DeclareMathOperator{\1}{\textbf{1}}
\DeclareMathOperator{\0}{\textbf{0}}

\newcommand{\defeq}{\mathrel{\mathop:}=} 
\newcommand{\Pn}[0]{\mathcal{P}_n}

\newcommand{\St}[0]{\mathcal{S}}
\newcommand{\Sg}[0]{\mathcal{G}}
\newcommand{\X}[0]{\mathcal{X}}
\newcommand{\Z}[0]{\mathcal{Z}}

\usepackage{diagbox}

\usepackage{tikz}
\usetikzlibrary{calc}

\usepackage{amssymb} 
   \makeatletter
     \renewcommand\@make@capt@title[2]{%
      \@ifx@empty\float@link{\@firstofone}{\expandafter\href\expandafter{\float@link}}%
       {\textbf{#1}}\@caption@fignum@sep#2\quad}%
     \makeatother
\makeatletter 
\renewcommand{\fnum@figure}{\textbf{Figure~\thefigure}}

\usetikzlibrary[positioning]
\usetikzlibrary{patterns}

\usepackage{mathtools}

\usepackage{ulem} 
\usepackage{tikz}
\usetikzlibrary[positioning]
\usetikzlibrary{patterns}

\usepackage{centernot}
\usepackage{mathtools}
\usepackage{stmaryrd}

\definecolor{violet}{HTML}{53257F} 
\definecolor{green}{HTML}{257a7f}
\definecolor{green2}{HTML}{527f27}
\definecolor{brown}{HTML}{852e29}

\colorlet{nodecolor}{black}
\colorlet{fadeoutnodecolor}{nodecolor!50}
\colorlet{edgecolor}{black}
\colorlet{fadeoutedgecolor}{edgecolor!50}
\colorlet{edgecolorU}{edgecolor}
\colorlet{edgecolorV}{fadeoutedgecolor}
\colorlet{edgecolorW}{fadeoutedgecolor}
\colorlet{circlecolor}{brown}

\usepackage{stackengine}
\newcommand{\xor}{\mathrel{\stackinset{c}{}{c}{-0.15ex}{$ \vee $}{$ \bigcirc $}}}

\newcommand\restr[2]{{
  \left.\kern-\nulldelimiterspace 
  #1 
  \littletaller 
  \right|_{#2} 
  }}
  
 \newcommand{\littletaller}{\mathchoice{\vphantom{\big|}}{}{}{}}

\usepackage{xparse}
\usetikzlibrary{matrix,backgrounds}
\pgfdeclarelayer{myback}
\pgfsetlayers{myback,background,main}
\tikzset{mycolor/.style = {line width=1bp,color=#1}}%
\tikzset{myfillcolor/.style = {draw,fill=#1}}%

\makeatletter
\def\blfootnote{\xdef\@thefnmark{}\@footnotetext}
\makeatother

\usepackage{accents}
\newlength{\dhatheight}

\usepackage{algorithm}
\usepackage[noend]{algpseudocode}

\begin{document}
\title{Algorithm to Verify Local Equivalence of Stabilizer States}
\author{\orcidlink{0000-0003-0418-257X}~Adam Burchardt}
\email[]{adam.burchardt.uam@gmail.com}
\address{QuSoft, CWI and University of Amsterdam, Science Park 123, 1098 XG Amsterdam, the Netherlands}
\author{Jarn de Jong}
\affiliation{Electrical Engineering and Computer Science Department, Technische Universit{\"a}t Berlin, 10587 Berlin, Germany}
\author{Lina Vandr\'e}
\affiliation{Naturwissenschaftlich-Technische Fakult\"at, Universit\"at Siegen, Walter-Flex-Stra\ss e 3, 57068 Siegen, Germany}
\date{\today}

\begin{abstract}
\noindent
We present an algorithm for verifying the local unitary (LU) equivalence of graph and stabilizer states. Our approach reduces the problem to solving a system of linear equations in modular arithmetic. 
Furthermore, we demonstrate that any LU transformation between two graph states takes a specific form, naturally generalizing the class of local Clifford (LC) transformations. Lastly, using existing libraries, we verify that for up to $n=11$, the number of LU and LC orbits of stabilizer states is identical.
%
\end{abstract}

\maketitle

\section{Introduction}

Stabilizer states are a key class of pure multipartite states, which are powerful resources for various quantum tasks and are efficiently implementable \cite{nielsen_chuang_2010,gottesman1997stabilizer,PhysRevLett.91.107903,PhysRevA.68.022312,PhysRevA.67.022310}. A notable subclass of stabilizer states are graph states, which can be visually represented by graphs: vertices correspond to qubits, and edges represent controlled two-qubit operations used in their experimental realization on quantum computers~\cite{PhysRevA.69.062311,graphstatesIBMQ,9866745,briegel2009measurement}. 


Entanglement properties are independent of the local basis which motivates to study local unitary (LU) equivalence of states. LU equivalent states are useful for exactly the same tasks. 
For graph states it has turned out that two different graphs may describe the same state up to LU operations \cite{vandennestGraphicalDescriptionAction2004}.

A natural subclass of LU operations on both stabilizer and graph states are local Clifford (LC) operations, which are fault-tolerant and easy to implement, making them well-suited for current quantum hardware~\cite{vandennestGraphicalDescriptionAction2004,Clifford,ketkarNonbinaryStabilizerCodes2006,BurchardtRaissi20,RaissiBurchardt22}. Additionally, there is an efficient algorithm for verifying LC equivalence between two stabilizer states, with a complexity of $\mathcal{O}(n^4)$ in the number of qubits $n$~\cite{BOUCHET199375,AlgorithmLULC,Bouchet1991AnEA}.

For a long time, it was conjectured that all local unitary (LU) transformations on stabilizer states could be implemented via LC operations \cite{LULCsupport,PLUvsLC}. This hypothesis, known as the  ``LU=LC conjecture'', represents a key open problem concerning the entanglement structure of stabilizer states \cite{krueger2005openproblemsquantuminformation}. Early work supported the conjecture for large subclasses of stabilizer states and linked it to quadratic forms~\cite{PhysRevA.71.022310}. It was shown to hold for graph states of up to 9 qubits \cite{cabelloEntanglementEightqubitGraph2009,Lina2024}. However, a numerical search uncovered a counterexample involving 27 qubits, which remains the smallest known example of such graphs to date \cite{LULCfalse}. Subsequently, a more systematic approach for constructing pairs of graph states that are LU-equivalent but not LC-equivalent was proposed~\cite{LUnoLC,claudet2024localequivalencestabilizerstates}. 


\begin{figure}[h]
\centering
\includegraphics[width=0.49\textwidth]{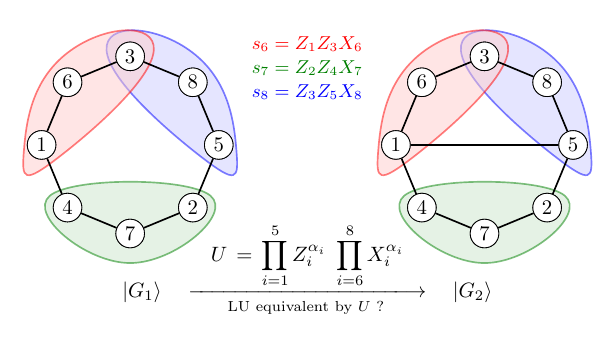}
\caption{Examples of two graph states that are not LU-equivalent. There exists an efficient algorithm exists to verify LC-equivalence between any two graph states. This paper introduces an algorithm to verify LU-equivalence between graph states, which appears efficient in practice. Using this algorithm, we verify that for $n \leq 11$, no instances of LU- but not LC-equivalent graph states are found. The algorithm relies on the following observations: using minimal local sets and corresponding stabilizers, we show that any LU-equivalence must have a very specific form. Subsequently, verification of the existence of such an equivalence can be reduced to solving a linear system of equations.}
\label{fig:intro}
\end{figure}

In this paper, we introduce an algorithm to verify LU equivalence between any two stabilizer states. Since any stabilizer state is LC equivalent to a graph state \cite{vandennestGraphicalDescriptionAction2004}, we can restrict our problem to determining LU equivalence between graph states. Our algorithm is constructive: for two input graphs $G$ and $G'$, it either outputs the exact form of the LU equivalence between the corresponding graph states $\ket{G}$ and $\ket{G'}$, or returns ``NO'' if the states are not LU equivalent. It builds upon several recent results concerning graph states \cite{LocalSets1,foliage,claudet2024covering,Lina2024}. While the algorithm appears borderline efficient, we do not yet have a formal proof of this statement. 

As a consequence, we show that any countrexample to ``LU=LC conjecture'' must be of a specific form. 
Moreover, we show that the LU transformations between two LU equivalent graph states take a specific form that naturally generalizes the class of local Clifford (LC) operations. This result paves the way for fault-tolerant implementations of LU equivalences for both graph and stabilizer states. 
Moreover, by utilizing an existing database containing all unlabeled LC-orbits for graphs with up to $n=11$ vertices, we applied our algorithm to graphs from different orbits and verified that graphs from distinct LC-orbits are not LU-equivalent. Consequently, we confirm that the number of LC-orbits and LU-orbits of unlabeled graphs is identical for $n\leq 11$ vertices.


The paper is organized as follows:
\cref{sec:Preliminaries} recalls the concept of graph states, stabilizer states, and hypergraph states, and it presents prior results concerning their equivalence. 
Furthermore, it introduces the concept of minimal local sets. 
\cref{sec:Necessary conditions} further investigates notions of local sets and based on them establish some necessary conditions for two graph states being LU equivalent. 
Those conditions are further used in \cref{sec:StepI}, which reduces the problem of finding an LU transformation between graph states related to arbitrary graphs to a simpler problem where  the initial graphs and the form of LU transformation are very specific.
\cref{sec:StepII} shows how to bring such a reduced problem to a linear system of equations in modular arithmetic. 
Solvability and further simplification of such a system are discussed in greater detail in \cref{AppB}. 
\cref{Algorithm_outline} connects results obtained in previous sections and outlines the algorithm that verifies LU equivalence between graph states. 
\cref{sec:numerics} contains a numerical analysis of LU equivalent graph states up to $n=11$ qubits. 
\cref{Discussion} contains concluding remarks and open questions. 

\section{Stabilizer states and Graph states}
\label{sec:Preliminaries}


In this section we present a short introduction of concepts relevant for this paper and fix notation. The introduction includes stabilizer, graph, and hypergraph states, as well as local Clifford equivalence of states, entanglement properties of graph states, marginal states, and the notion of local sets.

\subsection{Stabilizer states}
\noindent

An interesting family of multipartite entangled states are \textit{stabilizer states}. These are states which are +1 eigenstates of operators of the \textit{Pauli group}.
Recall that the $n$-th \textit{Pauli group}
\begin{equation}
\label{eq:PauliGroup}
{\displaystyle \mathbf{P}_{n}=\left\{{\begin{array}{c}\phi \cdot A_{1}\otimes \ldots \otimes A_{n}:
\,
A_{j}\in \mathbf{P} ,\ 
\,
\phi=\pm 1,\pm i \end{array}}\right\}}
\end{equation}
consists of $n$-fold tensor products of Pauli operators, i.e. $\mathbf{P}~=~\{\Id,X,Y,Z\}$, possibly multiplied by a scalar. The scalar $\phi$ is called the \textit{phase} of a Pauli string. By definition, the Pauli group is a subgroup of the group of local unitary operators $U_2^{\otimes n}$.

The \textit{stabilizer group} $\St_{\ket{\psi}}$ of an $n$-qubit state $\ket{\psi}$ is the group of Pauli operators $S\in\mathcal{P}_{n}$ of which $\ket{\psi}$ is $+1$ eigenstate. 
For any state, the stabilizer group $\St_{\ket{\psi}}$ is of order $|\St_{\ket{\psi}}|=2^k$ for some $k=0,\ldots,n$, and is generated by $k$ elements. 
If, for a given state $\ket{\psi}$, the number $k=n$, the state $\ket{\psi}$ is called \textit{stabilizer state}. 
For such states, the density matrix has a convenient expression in terms of their stabilizers:
\begin{equation}
\label{eq:stab_state_in_stabilizer}
\dyad{\psi} =
\dfrac{1}{2^n}
\sum_{S\in \St_{\ket{\psi}}} S.
\end{equation} 
The stabilizer group $\St_{\ket{\psi}}$ of a stabilizer state $\ket{\psi}$ is generated by $n$ elements $ \St_{\ket{\psi}}=\langle s_1,\ldots,s_n\rangle$ for some $s_i\in \St_{\ket{\psi}}$. We denote by $\Sg = \{ s_1,\ldots,s_n\}$ the \textit{set of generators} of stabilizer $\St_{\ket{\psi}}$. Notice that the choice of generators $\Sg$ of a stabilizer set is not unique. 

For example, the stabilizer group of the 3 qubit GHZ state $\ket{GHZ_3} = \tfrac{1}{\sqrt{2}}(\ket{000}+\ket{111})$ \cite{greenberger2007going} is generated by three independent elements:
\begin{align}
    \St_{\ket{GHZ_3}} &= \langle X^{(1)} X^{(2)} X^{(3)}, Z^{(1)} Z^{(2)} \id^{(3)}, Z^{(1)} \id^{(2)} Z^{(3)} \rangle  \notag \\ 
    &\eqqcolon \langle \Sg \rangle \label{eq:GHZstab}
\end{align}
and hence consists of $2^3$ elements. Here we use the notation $X^{(1)} X^{(2)} X^{(3)} :=X^{(1)}\otimes X^{(2)} \otimes X^{(3)}$. The GHZ state is a stabilizer state for all number of qubits $n$.

Recall that any stabilizer generator $\Sg=\langle s_1,\ldots,s_n\rangle$  can be associated to the check matrix, an $n\times 2n$ matrix whose rows correspond to the generators $ s_1$ through $s_n$, the left hand side of the matrix contains $1$s to indicate which generators contain $X$ on given position, and the right hand side contains which generators contain $Z$ on given position, the presence of a $1$ on both sides indicates a $Y$ in the
generator. More explicitly, the check matrix is of the following form
\[
\begin{bNiceArray}{c|c}
\X_\Sg&\Z_\Sg \\
\end{bNiceArray}
\]
where matrices $\X_\Sg=(x_{ij})_{i,j=1}^n$ and $\Z_\Sg=(z_{ij})_{i,j=1}^n$ are binary matrices defined by
\begin{equation*}
x_{ij}\defeq \begin{cases}
  1  & s_{i\downarrow j} =X, \, Y \\
  0 & \text{otherwise}
\end{cases},
\quad
z_{ij}\defeq \begin{cases}
  1  & s_{i\downarrow j} =Z, \, Y \\
  0 & \text{otherwise}
\end{cases},
\end{equation*}
where $s_{i\downarrow j}$ denotes $j$th element of operator $s_i$. For example the check matrix of the generator $\Sg$ given in \cref{eq:GHZstab} is given by
\begin{equation}
\begin{bNiceArray}{c|c}
\X_\Sg&\Z_\Sg \\
\end{bNiceArray}
=
\begin{bNiceArray}{c c c |c c c}
1 & 1 & 1 & 0 & 0 & 0 \\
0 & 0 & 0 & 1 & 1 & 0 \\
0 & 0 & 0 & 1 & 0 & 1
\end{bNiceArray}. \label{eq:GHZcheckM}
\end{equation}

If the check-matrix $[\X_\Sg | \Z_\Sg ]$ is in the following from:
\begin{equation}
\label{eq:stabilizer_generator_normal_form}
\X_\Sg=
\begin{bNiceArray}{c|c}
\Block{1-1}{\id_k}&\Gamma_{k|n-k}^{}\\
\cline{1-2}
\Block{1-1}{0}&\Block{1-1}{0}\\
\end{bNiceArray},
\quad
\Z_\Sg=
\begin{bNiceArray}{c|c}
\Block{1-1}{\Gamma_{k|k}}&\Block{1-1}{0}\\
\cline{1-2}
\Block{1-1}{\Gamma_{k|n-k}^{\text{T}}}&\Block{1-1}{\id_{n-k}}\\
\end{bNiceArray},
\end{equation}
for some symmetric $k\times k$ matrix $\Gamma_{k|k}^{}$ and arbitrary $k\times n-k$ matrix $\Gamma_{k|n-k}^{}$, we say that it is in a \textit{normal form}. 
If the symmetric matrix $\Gamma_{k|k}^{}$ does not have diagonal elements, we say that it is in a \textit{strong normal form}. 
The check matrix of the GHZ state in \cref{eq:GHZcheckM} is in a strong normal form, since  we have $\Gamma_{k|k}^{} = [0]$ and $\Gamma_{k|n-k}^{} = [1 \  1 ]$. 
Up to a permutation of the qubits, for any stabilizer one can always find a set of generators such that the associated check matrix is in a normal form, see \cref{appendix:A} for details. 
Recall that any stabilizer generator $\Sg=\langle s_1,\ldots,s_n\rangle$ is uniquely determined by the corresponding check matrix and the \textit{phase vector} $\phi=(\phi_j)_{j=1}^n$, where $\phi_j=\pm 1,\pm i$ is a scalar factor in a Pauli string $s_j$.  

\subsection{Local Clifford equivalence}

Stabilizer states are usually considered under local unitary or local Clifford equivalence~\cite{nielsen_chuang_2010}. 
The \textit{local Clifford group}
\begin{equation}
{\displaystyle \mathbf {C} _{n}^{\text{loc}}=\{V\in U_{2}^{\otimes n}\mid V\mathbf {P} _{n}V^{\dagger }=\mathbf {P} _{n}\}}
\end{equation}
is the local normalizer of the $n$-th Pauli group, i.e. the largest group preserving the Pauli group $\mathbf {P} _{n}$ that is also local. 

For example, the single-qubit Clifford group $\mathbf{C}_1$ has exactly $24$ elements and is generated by the Hadamard gate $H=1/\sqrt{2}\begin{psmallmatrix}1 & 1\\1 & -1\end{psmallmatrix}$, the phase gate $S=\begin{psmallmatrix}1 & 0\\0 & i\end{psmallmatrix}$ together with the Pauli $X$ matrix. In fact, each element in $\mathbf{C}_1$ can be uniquely expressed as a matrix product $\mathbf{A}\mathbf{B}$, where $\mathbf{A}\in\{I,H,S,HS,SH,HSH\}$ and $\mathbf{B}=\{\id,X,Y,Z\}$. Similarly, $|\mathbf{C}_n|=24^n$. 
By definition, the local Clifford group  $\mathbf {C} _{n}^{\text{loc}}$ is a subgroup of the group of local unitary operators $U_{2}^{\otimes n}$.
Two stabilizer states are said to be \textit{LC equivalent} (locally Clifford equivalent) if they can be transformed into each other using only local Clifford operators $\mathbf {C} _{n}^{\text{loc}}$. 
Similarly, two stabilizer states are said to be \textit{LU equivalent} (locally unitary equivalent) if they can be transformed into each other using only local unitary operators $U_{2}^{\otimes n}$. 


For the purposes of this manuscript, it is useful to distinguish a special classes of matrices. 
Firstly, observe that for any two Pauli matrices $A,B\in \{X,Y,Z\}$, there exists a matrix, referred to as the \textit{transition matrix} and denoted by $\mathbf{C} (A,B) $ that transforms $A$ into $B$ and $B$ into $A$ while preserving other Pauli matrices up to the phase. For instance, $\mathbf{C} (X,Z)=H$ as the Hadamard matrix $H$ satisfies: $HXH^\dagger=Z, \, HZH^\dagger=X,\, HYH^\dagger=-Y$. We refer to \cref{appendix:0} for the exact form and some simple observations regarding transition matrices. 

Furthermore, for Pauli matrices, we define their powers in the following way:
\begin{align}
X^\alpha &:= e^{i \tfrac{\alpha}{2}\pi} \big( \text{cos} \tfrac{\alpha}{2}\pi \,\Id -i \,\text{sin} \tfrac{\alpha}{2} \pi \, X \Big)
\nonumber 
\\
Y^\alpha &:= e^{i \tfrac{\alpha}{2}\pi} \big( \text{cos} \tfrac{\alpha}{2}\pi \,\Id -i \,\text{sin} \tfrac{\alpha}{2} \pi \, Y \Big)
\label{eq:XYZ_alpha}
\\
Z^\alpha &:= e^{i \tfrac{\alpha}{2}\pi} \big( \text{cos} \tfrac{\alpha}{2}\pi \,\Id -i \,\text{sin} \tfrac{\alpha}{2} \pi \, Z \Big)
\nonumber 
\end{align}
Notice that for $\alpha=0$, we have $X^0=Y^0=Z^0=\Id$. For $\alpha=1$, we recover initial Pauli matrices, i.e. $X^1=X$ and similarly for $Y,Z$. Furthermore, the powers satisfy the periodicity property $X^{\alpha +2}=X^{\alpha}$ and similarly for $Y,Z$. Additionally, observe that 
$X^{\alpha},Y^{\alpha},Z^{\alpha}$ are Clifford matrices only for $\alpha=k/2$ for arbitrary $k\in \mathbb{Z}$.

\subsection{Graph states}


A graph $G=(V,E)$ consists of a finite set of vertices $V$ and of edges $E\subseteq V \times V$. A graph is called \textit{simple} if it contains neither edges connecting a vertex to itself nor multiple edges between the same vertices. A simple graph is uniquely characterized by its \textit{adjacency matrix} $\Gamma_G =(\gamma_{ij})_{i,j \in V}$, where $\gamma_{ij}=1$ if vertices $i,j$ are connected, and $\gamma_{ij}=0$ otherwise. For simplicity of notation, we consider graphs with the vertex set $V=[n]:=\{1,\ldots ,n\}$. The \textit{neighbourhood} $N_i:=\{k\in [n]: \gamma_{ik}=1\}$ of a vertex $i$ is the set of all vertices adjacent to $i$.

The \textit{graph state} $\ket{G}$ corresponding to a simple graph $G$ is defined as
\begin{equation}
\label{graph_state}
\ket{G} = \prod_{(v,w)\in E} \textbf{CZ}^{\{vw\}} \ket{+}^{\otimes \abs{V}},
\end{equation}
where $\textbf{CZ}^{\{vw\}} = \sum_{i,j =0}^1  (-1)^{ij} \ket{ij}_{vw} \bra{ij}$ is a controlled-$Z$ operator on qubits $v$ and $w$, and $\ket{+}= \frac{1}{\sqrt{2}} (\ket{0}+\ket{1})$. 
In particular, graph states are stabilizer states, with stabilizer generators of the following form:
\begin{equation}
\label{eq:stabilizer_of_graph_state}
s_i=
X^{(i)}\cdot\prod_{k\in N_i}Z^{(k)},
\end{equation}
for $1\leq i\leq n$ and where $A^{(i)}\in \Pn$ denotes operator $A$ acting on the $i$-th position, i.e. $A^{(i)}=\Id^{\otimes i-1}\otimes A \otimes \Id^{\otimes n-i}$, and $N_i$ denotes a neigbourhood of vertex $i$. 
Notice that the standard generators of a graph state $\ket{G}$ correspond to the following check matrix
\[
\begin{bNiceArray}{c|c}
\X_\Sg&\Z_\Sg \\
\end{bNiceArray}
=
\begin{bNiceArray}{c|c}
\Id_n&\Gamma_G \\
\end{bNiceArray},
\]
where $\Gamma_G$ is the adjacency matrix corresponding to $G$. 
Notice that this is a strong normal form (compare \cref{eq:stabilizer_generator_normal_form}) of the check matrix. 

\begin{figure}
    \centering
    \includegraphics[width=0.49\textwidth]{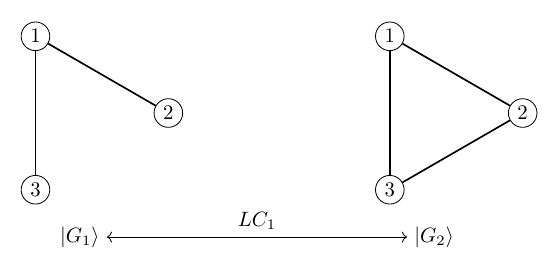}
    \caption{The 3 vertex star graph and the fully connected graph. Both corresponding graph states are LU equivalent to the 3 qubit GHZ state. The graphs are equivalent up to local complementation in vertex 1.}
    \label{fig:GHZ_example}
\end{figure}

A prominent example of graph states is the $n$ qubit GHZ state which can be represented by a star graph, as shown in \cref{fig:GHZ_example}. 
The graph state associated to the 3 qubit star graph is given by 
\begin{align}
\nonumber
    \ket{G_1} 
    &= \textbf{CZ}^{\{1,2\}}\textbf{CZ}^{\{1,3\}}\ket{+}^{\otimes 3} \\
    &= \id^{(1)} H^{(2)} H^{(3)} \ket{GHZ_3}.
\end{align}
The stabilizer generators are given by
\begin{align}
    s_1 &= X^{(1)} Z^{(2)} Z^{(3)}, \nonumber \\
    s_2 &= Z^{(1)} X^{(2)} \id^{(3)}, \\
    s_3 &= Z^{(1)} \id^{(2)} X^{(3)}. \nonumber 
\end{align}

It is known that any stabilizer state is LC equivalent to a graph state \cite{vandennestGraphicalDescriptionAction2004}. 
Furthermore, two graph states are LC equivalent if and only if their associated graphs are equivalent under \textit{local complementation operations} \cite{vandennestGraphicalDescriptionAction2004}.

\begin{definition}[Local complementation]\label{def:local_complementation} 
 For every vertex $a \in V$ of a graph $G=(V, E)$, we define a \textit{locally complemented graph} $\tau_{a}(G)$ with adjacency matrix
	\begin{equation}\label{eq:local_complementation_adjacency}
	\Gamma_{\tau_{a}(G)} = \Gamma_{G}+\Gamma_{K_{N_{a}}} \quad (\bmod 2),
	\end{equation}
	where $K_{N_{a}}= (V, {\{(v,w)~\mid~\left(v,w\in N_{a}\right) \wedge \left(v \neq w\right)\}})$ 
	is the complete graph on the neighborhood $N_{a}$ of vertex $a$.
\end{definition}

\noindent
For example the two graphs in \cref{fig:GHZ_example} are equivalent under local complementation in vertex 1. Therefore both corresponding graph states are equivalent to the GHZ state.

\begin{theorem}[see~\cite{vandennestGraphicalDescriptionAction2004}]
\label{LCisLC}
Two graph states are LC equivalent if and only if their associated graphs are equivalent under a sequence of local complementations.
\end{theorem}

\noindent
We use the acronym LC equivalent both for the local Clifford equivalence and for the local complementation equivalence of the related graphs due to \cref{LCisLC}.

\subsection{Entanglement of Graph states}

In early studies of graph states, it was observed that the entanglement properties of a graph state are completely characterized by its corresponding adjacency matrix \cite{PhysRevA.69.062311,Entanglement_GraphStates,HajdusekMurao_conference}. 
For a given graph $G$ represented by its adjacency matrix $\Gamma$, we define the \textit{Cut-rank} function on all subsets $M\subset [n]$ as
\begin{equation}
\CutRank : M     \mapsto \rk (\Gamma_{M \vert M^{\perp}}),
\end{equation}
where $\Gamma_{M \vert M^{\perp}}$ denotes the submatrix of the adjacency matrix $\Gamma$ with rows and columns indexed by subsets $M$ and $M^{\perp}:= [n]\setminus M$ respectively, which is also known as \textit{connectivity function} \cite{doi:10.1137/0608028,BOUCHET199375}. 
The Cut-rank function captures the entanglement properties of a graph state with respect to any bipartition. 
Indeed, recall that the entanglement entropy $E_M (\ket{\psi})$ with respect to the bipartition $M \vert M^{\perp}$ is defined as $E_M (\ket{\psi}) := - \tr (\rho_M \log_2 \rho_M)$, which for graph states is equal to $\CutRank (M)$, i.e. 
\begin{equation}
E_M (\ket{\psi})=\CutRank  (M)
.
\end{equation} 
As entanglement entropy does not change with LU operations, the graphs $G$ and $G'$ related to LU equivalent graph states have the same $\CutRank$ function. We say that two graph states $\ket{G}$ and $\ket{G'}$ are marginal equivalent if the corresponding graphs have the same $\CutRank$ function.

We compute the Cut-rank function of the graph state $\ket{G_1}$ in \cref{fig:intro} for some sets $M \subset [8]$ as an example: 
$ \CutRank(\{1, 6\})= 2$, 
$\CutRank(\{1,3,6\}) =2$, 
$ \CutRank(\{1,3,5,7\}) = 3$.


\subsection{Marginal States}

For a given stabilizer state $\ket{\psi}$ associated to a stabilizer group $\St_{\ket{\psi}}$ and any subset $M \subset [n]$, we define the \textit{reduced stabilizer group}: 
\begin{equation}
\label{eq:marginalstab}
\St_{\ket{\psi}}^M 
= \lbrace s \in \St_{\ket{\psi}} : \supp (s) \subset M \rbrace,  
\end{equation}
where $\supp  (s)$ denotes the \textit{support} of a Pauli operator $s$, i.e. the collection of tensor-subspaces on which it acts non-trivially. 
Notice that $\St_{\ket{\psi}}^M $ forms a subgroup of $\St_{\ket{\psi}}$. 
This subgroup determine the density matrix of a reduced state $\rho_M (\psi ):= \tr_{M^\perp} \ket{\psi}\bra{\psi}$, namely:
\begin{equation}
\label{eq:marginalstab_def}
\rho_M (\psi ) =\dfrac{1}{2^{|M|}}
\sum_{s\in \St_{\ket{\psi}}^M} s_{\downarrow M} ,
\end{equation}
where $s_{\downarrow M}$ denotes the operator $s$ restricted to the subsystems in $M$. 
Note the similarity to \cref{eq:stab_state_in_stabilizer}. 
Furthermore, the size of the group $\St_{\ket{\psi}}^M $ determines the entanglement entropy, namely
\begin{equation}
\label{eq:ST_M_and_CutRank}
\Big|\St_{\ket{\psi}}^M \Big| =2^{|M|-\CutRank (M)}
\, .
\end{equation}

For example the reduced stabilizer group $\St_{\ket{G_1}}^{\{1,3,6\}}$ of the graph state $\ket{G_1}$ in \cref{fig:intro} is given by
\begin{align}
    \St_{\ket{G_1}}^{\{1,3,6\}} = \{ \id^{(M)}, Z^{(1)} Z^{(3)} X^{(6)} \id^{(M\setminus  \{1,3,6\} )} \}.
\end{align}
Therefore, the corresponding reduced state is determined by
\begin{align}
    \rho_{\{1,3,6\}} (G_1 ) =\dfrac{1}{2^{2}} \left( \id^{(1,3,6)} + Z^{(1)} Z^{(3)} X^{(6)} \right).
\end{align}

\subsection{Local sets}
\label{subsec:localset}

Consider graph $G$ and any subset $M\subset [n]$ of its vertices. 
Following Ref~\cite{LocalSets1}, a set $M$ is called a \textit{local set} if it is of the form $M=D\cup \Odd{D}$ where 
\begin{equation}
\label{Eq:OddD}
\Odd{D}:=\Big\{v\in [n]: |N_v \cap D| \text{ is odd}\Big\} 
\end{equation}
is the set of vertices connected to an odd
number of vertices in $D$. 
The set $D$ is called a \textit{generator} of the local set $M$ \cite{LocalSets1}. 
A local set that is minimal by inclusion relations is called a \textit{minimal local set} (MLS).  

For example the set $M_1 = \{3,5,8\}$ of graph $G_1$ in \cref{fig:intro} is a local set generated by $D_1 = \{ 8\}$. The set of neighboring vertices with an odd number of vertices connected to the set $D_1$ is $\Odd{D_1} = \{3,5 \}$. Another local set is $M_2 = \{2,3\}$ of graph $G_2$ in \cref{fig:GHZ_example}. In this case, $D_2 = M_2$ and $\Odd{D_2} = \emptyset$, since vertex 1 is adjacent to $D_2$ by an even number of edges. Note that  both local sets are MLSs.

Notice that a MLS can be characterized in terms of the $\CutRank$ function. 
Indeed, a set $M$ is MLS if and only if $\CutRank (M) <|M|$ and $\CutRank (N) =|N|$ for any proper subset $N\subsetneq M$, see the recent results in Ref~\cite{claudet2024covering,Lina2024}. 
Similarly,  MLS can be characterized in terms of the reduced stabilizer group, see \cref{eq:ST_M_and_CutRank}. 
Indeed, a set $M$ is MLS if and only if $\St_{\ket{G}}^M \neq \{\id^n\}$, while $\St_{\ket{G}}^N = \{\id^n\}$ for any proper subset $N\subsetneq M$. 
It turned out that the reduced stabilizer group $\St_{\ket{G}}^M$ corresponding to MLS $M$ has either two or for elements, see \cite{claudet2024covering,Lina2024}. 
Therefore, we introduce the \textit{Type} of an MLS as follows. 

\begin{definition}
\label{lemma:types}
Consider MLS $M$ and the corresponding reduced stabilizer group $\St_{\ket{\psi}}^M$. 
We call $M$ the \textit{MLS of Type I} if $|\St_{\ket{G}}^M |=2$, and \textit{MLS of Type II} if $|\St_{\ket{G}}^M |=4$.
\end{definition}

\noindent
A collection $\mathcal{M}$ of sets is called a minimal local set cover (MLS cover) if each set $M\in \mathcal{M}$ is a local set and $\mathcal{M}$ covers all vertices of a graph, i.e. $\cup_{M\in \mathcal{M}} M = [n]$. 
It was recently shown that such a cover exists for all graphs:

\begin{theorem}[see Theorem 1 in \cite{claudet2024covering}]
\label{th:MLS_cover}
Any graph has an MLS cover. 
Furthermore, there is an algorithm to compute an MLS cover for any graph, that runs in $\mathcal{O}(n^4)$ in the number of graph vertices $n$. 
\end{theorem}

For example, a MLS cover for the graphs $G_1,G_2$ in \cref{fig:intro} is given by 
\begin{align}
    \mathcal{M} = \{ \{1,3,6\}, \{3,5,8\},\{2,4,7\} \}.
\end{align}
Note that this choice is not unique.

As MLSs are uniquely characterized by the properties of the $\CutRank$ function, and LU-equivalent states share the same $\CutRank$ function, it follows that they share the same MLSs.

\begin{proposition}
\label{prop:MLS_cover_same}
If two graph states $\ket{G}$ and $\ket{G'}$ are LU equivalent, then any MLS for $G$ is also the MLS for $G'$ of the same Type, see \cref{lemma:types}. In particular, any MLS cover for $G$ is simultaneously an MLS cover for $G'$.    
\end{proposition}

\subsection{(Weighted) Hypergraph states}


We conclude this section by generalizing the notion of graph states to hypergraph states, and further to weighted hypergraph states \cite{Kruszynska_2009, Qu_2013, Hypergraphs_2013, Guehne_2014, Mari_2019_Dissertation}. 
A hypergraph $H=(V,E)$ consists of a finite set of vertices $V$ and of a set of hyperedges $E =\{e : e\subset V\}$. 
Notice that if each $e\in E$ is of the size $|e|=2$, then a hypergraph is a simple graph. 
For a given set $D$, a $D$-weighted hypergraph $H_f=(V,f:2^V\rightarrow D)$ consists of a finite set of vertices $V$ and a function $f:2^V\rightarrow D$ defined on all subsets of $V$ with values in $D$. 
Notice that we can think of a hypergraph as a $\{0,1\}$-weighted hypergraph with hyperedges set given by $E=\{e\subset V :f(e)=1\}$. 

For an $\mathbb{R}$-weighted hypergraph $H_w=(V,f:2^V\rightarrow \mathbb{R})$, we define a \textit{weighted hypergraph state} as 
\begin{equation}
\label{weighted_hypergraph_state}
\ket{H_w} = \prod_{e\subset V} \textbf{CZ}^{e,f(e)} \ket{+}^{\otimes V},
\end{equation}
where for a hyperedge $e=\{v_1,\ldots, v_k\}$, 
\begin{align}
\label{eq:CZ_hypergraph}
\textbf{CZ}^{e,\alpha} = \sum_{j_1, \dots, j_k \in \{0,1\}} e^{i \pi  \alpha j_1\cdots j_k} \dyad{j_1,\ldots, j_k}_{e}
\end{align}
is a controlled-$Z^\alpha$ operator acting on qubits in $e$, and $\ket{+}= \frac{1}{\sqrt{2}} (\ket{0}+\ket{1})$ \cite{Hypergraphs_2013,zakaryan2024nonsymmetricghzstatesweighted}. 
Notice that this formula simplifies to \cref{graph_state} for graph states. 
Furthermore, weighted hypergraph states coincide with locally maximally entangleable  states \cite{Kruszynska_2009}. 



\section{Necessary conditions for LU equivalence}
\label{sec:Necessary conditions}

We begin with the statement that for a MLS $M$ any LU operator must transform between elements in the respective local stabilizer groups. 

\begin{lemma}
\label{prop:necessary_conditions_MLS}
Let $M\subset [n]$ be a MLS of the graphs $G$ and $G'$. 
Suppose that LU operator $U=U_1\otimes\cdots\otimes U_n$ transforms $\ket{G}$ into $\ket{G'}$. 
Then, there is a bijection $f_M: \St_{\ket{G}}^M \rightarrow \St_{\ket{G'}}^M$ such that
\begin{align}
\label{eq:s_sprim}
f_M (s) = U s U^\dagger 
\end{align}
for arbitrary $s\in \St_{\ket{G}}^M$ and that preserves identity, i.e. $f(\id^{\otimes n})=\id^{\otimes n}$. Moreover, if $M$ is MLS of Type II, $U_i$ must be a Clifford matrix for any $i\in M$.
\end{lemma}

\noindent
\cref{prop:necessary_conditions_MLS} is proven in \cref{app:MLS}. Note that, in general, LU equivalence between two graph states does not map a stabilizer of one state to the corresponding stabilizer of the other. Thus, a similar result for local sets or other types of sets does not hold. This makes MLSs particularly useful, as they have very few elements, and any LU operator maps corresponding stabilizers into each other. 
In fact, \cref{prop:necessary_conditions_MLS} can be directly applied to exclude the possibility of LU-equivalence between two graph states, as demonstrated by the following example.

\begin{example}[\cite{Lina2024}]
\label{ex:9_qubit}
Consider two graphs on $9$ vertices presented in \cref{fig:example_9qubit}. 
One can verify that both graphs have the same $\CutRank$ function, hence the related graph state $\ket{R}$ and $\ket{L}$ are marginal-equivalent. 
On the other hand, the Bouchet algorithm shows, that this two graphs are not LC equivalent~\cite{BOUCHET199375,AlgorithmLULC}. 
We use \cref{prop:necessary_conditions_MLS} to show that those two graphs are not LU equivalent. 

This can be done by considering any MLS cover that contains the sets $M=\{1,2,3,5\}$ and $M'=\{1,4,5,7\}$. 
Notice that both sets are of Type I, and corresponding reduced stabilizer groups contain the following unique non-trivial elements:
\begin{align*}
S^{M}_{\ket{L}} &\ni s_2^L = Z^{(1)} X^{(2)} Z^{(3)} Z^{(5)} ,\\
S^{M}_{\ket{R}} & \ni s_{1}^R s_{2}^R s_{3}^R s_{5}^R =  -Y^{(1)} Y^{(2)} Y^{(3)} Y^{(5)}, \\
S^{M'}_{\ket{L}} &\ni s_{4}^{L} = Z^{(1)} X^{(4)} Z^{(5)} Z^{(7)}, \\
S^{M'}_{\ket{R}} &\ni s_{4}^{R} = Z^{(1)} X^{(4)} Z^{(5)} Z^{(7)}.
\end{align*}
which by \cref{prop:necessary_conditions_MLS}, are transformed one to another, i.e. 
\begin{align}
    U \,(Z^{(1)} X^{(2)} Z^{(3)} Z^{(5)} ) \,U^\dagger &=-Y^{(1)} Y^{(2)} Y^{(3)} Y^{(5)}, \label{eq:9ex_cond1} \\
    U \,(Z^{(1)} X^{(4)} Z^{(5)} Z^{(7)} ) \,U^\dagger &=Z^{(1)} X^{(4)} Z^{(5)} Z^{(7)}, \label{eq:9ex_cond2} 
\end{align}
by any LU equivalence $U=U_1\otimes\cdots\otimes U_n$ between $\ket{L}$ and $\ket{R}$. Notice that this leads to contradiction, namely \cref{eq:9ex_cond1} indicates that $U_1 \,Z_1 \,U_1^\dagger =Y^{(1)}$ while \cref{eq:9ex_cond2} requires $U_1 \,Z^{(1)}  \,U_1^\dagger =Z^{(1)} $. The argument relies on the fact that the two MLSs $M$ and $M'$ intersect non-trivially on vertex $1\in M\cap M'$ and that reduced stabilizer groups of LU-equivalent states must be compatible.
\end{example}


\begin{figure}
\centering
\includegraphics[width = 0.8\linewidth]{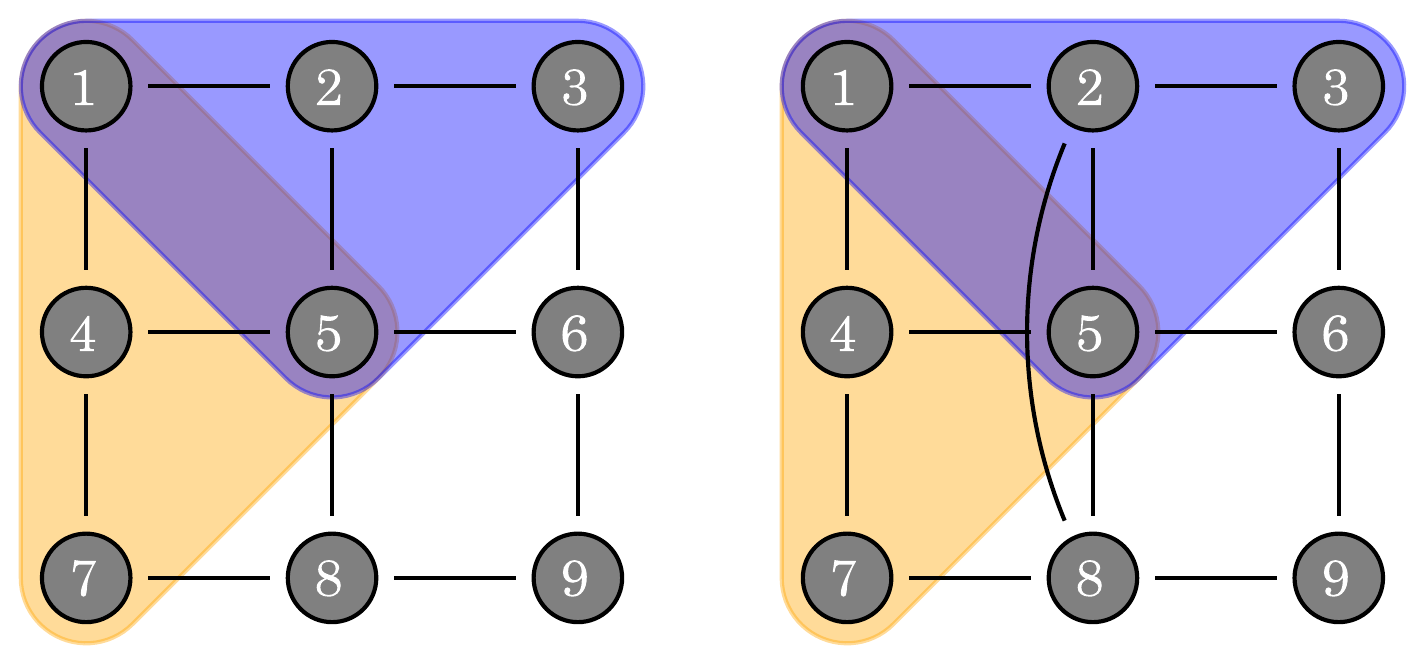}
\caption{Two graph $L$ (left) and $R$ (right) on nine vertices. Corresponding graph states are marginal-equivalent, but not LU-equivalent. Indeed, as illustrated in \cref{ex:9_qubit}, it is impossible to construct compatible functions corresponding to two MLSs: $M=\{1,2,3,5\}$ and $M'=\{1,4,5,7\}$, which are highlighted in blue and orange, respectively.
This figure is taken from Ref~\cite{Lina2024}.
}
\label{fig:example_9qubit}
\end{figure}

We now continue with the statement that any LU operator transforming two graph states must transform at least one Pauli string into another. Namely, we have the following statement.
\begin{proposition}
\label{prop:fG_fG'_functions}
Suppose that graph states $\ket{G}$ and $\ket{G'}$ are LU equivalent by operator $U=U_1\otimes\cdots\otimes U_n$. Then, there exist functions $F_G,F_{G'}:[n]\rightarrow \{X,Y,Z\}$ such that 
\begin{align}
\label{eq:fG_fG'_functions}
F_{G'} (i) = U _i \,F_G(i) \,U_i^\dagger 
\end{align}
for arbitrary $i\in [n]$. 
\end{proposition}

\begin{proof}
Consider a MLS cover $\mathcal{M}$ of graph $G$, see \cref{th:MLS_cover}). Notice that if $\ket{G'}$ is LU equivalent to $\ket{G}$, $\mathcal{M}$ is also MLS cover of $G'$, see \cref{prop:MLS_cover_same}. Consider arbitrary vertex $i \in[n]$. It belongs to at least one MLS $M\in \mathcal{M}$. Since any LU equivalence transforms stabilizer elements corresponding to $M$, at least one Pauli matrix must be transformed to another for vertex $i$, which gives values $F_G(i),F_{G'}(i)$ of the functions $F_G,F_{G'}$. 
\end{proof}

Consider the problem of determining whether two graphs $G,G'$ are LU equivalent. The task is to identify functions $F_G,F_{G'}:[n]\rightarrow \{X,Y,Z\}$ such that \cref{eq:fG_fG'_functions} in \cref{prop:fG_fG'_functions} holds. A straightforward approach is to find an MLS cover for $G$ and check whether it is also the MLS cover for $G'$.  If this is the case, each vertex $i\in [n]$ belongs to some set $i\in M$ from the MLS cover. When $M$ is of Type I, the reduced stabilizer groups $\St_{\ket{G}}^M, \St_{\ket{G,}}^M$ contain only a single non-trivial element, allowing us to determine the value of functions $F_G(i),F_{G'}(i)\in\{X,Y,Z\}$ based on \cref{prop:necessary_conditions_MLS}. However, if $M$ is of Type II, the reduced stablizer groups $\St_{\ket{G}}^M, \St_{\ket{G,}}^M$ contain three elements each, and the precise manner in which LU equivalence permutes these elements is not immediately clear. This ambiguity results in nine possible combinations for the values of $F_G(i),F_{G'}(i)\in\{X,Y,Z\}$. 

This issue can be mitigated by considering the \textit{intersection graph} of the MLS cover $M = \{M_i\}$, where each vertex corresponds to a set $M_i$, and two vertices $M_i$ and $M_j$ are connected by an edge if and only if $M_i \cap M_j \neq \emptyset$. Let $t$ denote the number of connected components in the intersection graph consisting solely of MLSs of Type II. By employing overlap-based arguments, as presented in \cref{ex:9_qubit}, we derive the following result:

\begin{proposition} 
\label{prop_pairs_of_func}
Consider two graphs $G$ and $ G'$ sharing the same MLS cover $\mathcal{M}$. 
There is a procedure that runs in $\mathcal{O}(n^3)$ time, which either rules out LU equivalence between $\ket{G}$ and $\ket{G'}$ or identifies pairs of functions $F_G^{\ell},F_{G'}^{\ell}:[n]\rightarrow \{X,Y,Z\}$ indexed by parameter $\ell\in \mathcal{L}$, such that if $\ket{G},\ket{G'}$ are LU-equivalent by some $U=U_1\otimes\cdots\otimes U_n$, then 
    \begin{align}
\label{eq:fG_fG'_functions_only_if}
F_{G'}^\ell (i) = U _i \,F_{G}^\ell (i) \,U_i^\dagger  
\end{align}
for at least one $\ell\in \mathcal{L}$. Moreover, $|\mathcal{L}|\leq 3^t$ where $t$ denotes the number of connected components in the intersection graph of $\mathcal{M}$ consisting entirely of MLSs of Type II. 
\end{proposition}

\noindent
Algorithm 1 illustrates the procedure for determining pairs of functions $F_G^{\ell},F_{G'}^{\ell}$, and we prove its correctness in \cref{appendix:D}.

We conclude this section with a remark. The ambiguity in determining the functions $F_G^{\ell},F_{G'}^{\ell}$ arises specifically for vertices belonging to MLSs of Type II. However, as established in \cref{prop:necessary_conditions_MLS}, any LU operator acting on these vertices is necessarily Clifford. Nonetheless, it remains unclear how this observation could significantly simplify the algorithm presented in this paper.

\section{Reduction to special form of graphs and LU equivalence}
\label{sec:StepI} 

In the previous section, we showed that any LU equivalence between two graph states must map at least one Pauli string to another, see \cref{prop:fG_fG'_functions}. 
These Pauli strings were represented as functions $F_G,F_{G'}:[n]\rightarrow \{X,Y,Z\}$. While we could not directly identify a specific Pauli string, we introduced a procedure to find pairs of such strings/functions $F_G^\ell,F_{G'}^\ell$, guaranteeing that any LU equivalence must map at least one such pair. 
In this section, we demonstrate that if a specific pair $F_G^\ell,F_{G'}^\ell$ is mapped under LU equivalence, the initial graph states can be transformed into LC-equivalent representatives $|\hat{G}\rangle$ and $|\hat{G}\rangle$, where the LU equivalence between them takes a very specific form. Namely, we have the following results.

\begin{theorem}
\label{corollary:all_are_that_form}
    Any pair of LU-equivalent graph states $\ket{G}$ and $\ket{G'}$ is LC-equivalent to the pair $\hat{G}$ and $\hat{G'}$, i.e. 
    \[\ket{\hat{G}} \stackrel{LC}{\cong} \ket{G},
\quad
\ket{\hat{G'}} \stackrel{LC}{\cong} \ket{G'},\]
with adjacency matrices of the form
\begin{equation}
\label{auxiii_main_main}
\Gamma_{\hat{G}}=
\begin{bNiceArray}{c|c}
\Block{1-1}{\Gamma_{k|k}}&\Gamma_{k|n-k}^{}\\
\cline{1-2}
\Block{1-1}{\Gamma_{k|n-k}^{\text{T}}}&\Block{1-1}{0}\\
\end{bNiceArray},\quad
\Gamma_{\hat{G'}}=
\begin{bNiceArray}{c|c}
\Block{1-1}{\Gamma_{k|k}'}&\Gamma_{k|n-k}^{}\\
\cline{1-2}
\Block{1-1}{\Gamma_{k|n-k}^{\text{T}}}&\Block{1-1}{0}\\
\end{bNiceArray}.
\end{equation} 
which are LU equivalent, i.e. $|\hat{G'}\rangle=\hat{U}|\hat{G}\rangle$,  by the operator of the following form
\begin{equation}
\label{eq:ZX_main}
\hat{U}=
\prod_{i=1}^k Z^{\alpha_i}_i
\prod_{i=k+1}^nX^{\alpha_i}_i
\end{equation}
for some values $\alpha_i \in [0,2]$.
\end{theorem}

\begin{proposition}
\label{Th:reduction}
Consider two graph states $\ket{G}$ into $\ket{G'}$, and all LU operators $U=U_1\otimes\cdots\otimes U_n$ that satisfies 
\[
F_{G'} (i) = U _i \,F(i) \,U_i^\dagger
\]
for given functions $F_G,F_{G'}:[n]\rightarrow \{X,Y,Z\}$, and for all $i\in [n]$. 
Then, there is a procedure running in $\mathcal{O}(n^3)$ time, which either rules out LU equivalence between $\ket{G}$ and $\ket{G'}$ or constructs graphs $|\hat{G}\rangle$ and $|\hat{G'}\rangle$ as discussed in the previous theorem.   
\end{proposition}

\noindent
The construction of graphs $\hat{G}$ and $\hat{G'}$ is algorithmic and presented in Algorithm 2 on \cref{fig:algo_subroutines}. \cref{appendix:B} contains proof of \cref{corollary:all_are_that_form} and \cref{Th:reduction}.

Notice that pairs of graphs with adjacency matrices of the form (\ref{auxiii_main_main}) and LU equivalence of the form (\ref{eq:ZX_main}) were previously used in Ref.~\cite{LULCfalse,LUnoLC,claudet2024localequivalencestabilizerstates} to construct counterexamples to the ``LU=LC conjecture'', i.e. graph state pairs that are LU- but not LC-equivalent, see \cref{fig:K_graph}. 
\cref{corollary:all_are_that_form} establishes the reverse statement, namely that all such counterexamples must take this form. 

We conclude this section by listing some known examples related to the ``LU=LC conjecture'' from the literature. After analyzing these examples, we propose a conjecture.

\begin{example}[Structure of counterexamples to the ``LU=LC conjecture'']
\label{ex:LU_LC_counterexamples}
\label{ex:K_graph} 
A bipartite \textit{Kneser graph} $ K(k,t) $ is a bipartite graph with two independent sets of vertices. The first set is $V_1 := [k]$, and the second set consists of all subsets \( A \subset [k] \) of size \( t \), denoted as \( V_2 := \binom{[k]}{t} \). Two vertices \( i \in V_1 \) and \( A \in V_2 \) are connected if and only if \( i \in A \). The graph \( K(k,t) \) has \( |V_1| + |V_2| = k + \binom{k}{t} \) vertices. \Cref{fig:K_graph} presents an example of \( K(7,5) \) with 28 vertices. 
The class of bipartite Kneser graphs can be generalized as follows. The \textit{generalized bipartite Kneser graph} \( K(k, (t_1, \ldots, t_s)) \) is a bipartite graph where the first set of vertices is \( V_1 := [k] \), and the second set \( V_2 \) is composed of subsets of \( [k] \) of sizes \( t_1, \ldots, t_s \), i.e., \( V_2 := \binom{[k]}{t_1} \cup \cdots \cup \binom{[k]}{t_s} \). Two vertices \( i \in V_1 \) and \( A \in V_2 \) are connected if and only if \( i \in A \). 
%
We denote by \( \hat{K}(k,t)) \) a graph obtained from $ K(k, t) $ by introducing additional edges between all vertices in \( V_1 := [k] \), similarly \( \hat{K}(k, (t_1, \ldots, t_s)) \). See \Cref{fig:K_graph} for an illustration.

The first known counterexample to the LU-LC equivalence conjecture was obtained via numerical search in Ref.~\cite{LULCfalse}. Initially, its structure was unclear. However, as pointed out in Ref.~\cite{LUnoLC}, the pair of graphs from Ref.~\cite{LULCfalse} is equivalent to the pair \( K(6, (5, 4)) \) and \( \hat{K}(6, (5, 4)) \). Moreover, Ref.~\cite{LUnoLC} systematically used the pairs \( K(k, (t_1, \ldots, t_s)) \) and \( \hat{K}(k, (t_1, \ldots, t_s)) \) to construct additional counterexamples to the ``LU=LC conjecture''. \Cref{fig:K_graph} illustrates one such example. More recently, Ref.~\cite{claudet2024localequivalencestabilizerstates} established conditions on \( k \) and \( t \) for which the pair of graphs \( K(k,t) \) and \( \hat{K}(k,t) \) is LU- but not LC-equivalent.
\end{example}

\begin{conjecture}
    Any pair of graphs $G,G'$ which is LU, but not LC-equivalent can be transformed by LC-operatons to the pair of graphs $K(n,(k_1,\ldots,k_s))$ and $ \hat{K}(k, (t_1, \ldots, t_s)) $, up to the removal of some twin vertices (i.e. vertices with identical neighborhoods). 
\end{conjecture}

\begin{figure*}
    \centering
    \includegraphics[width=0.8\linewidth]{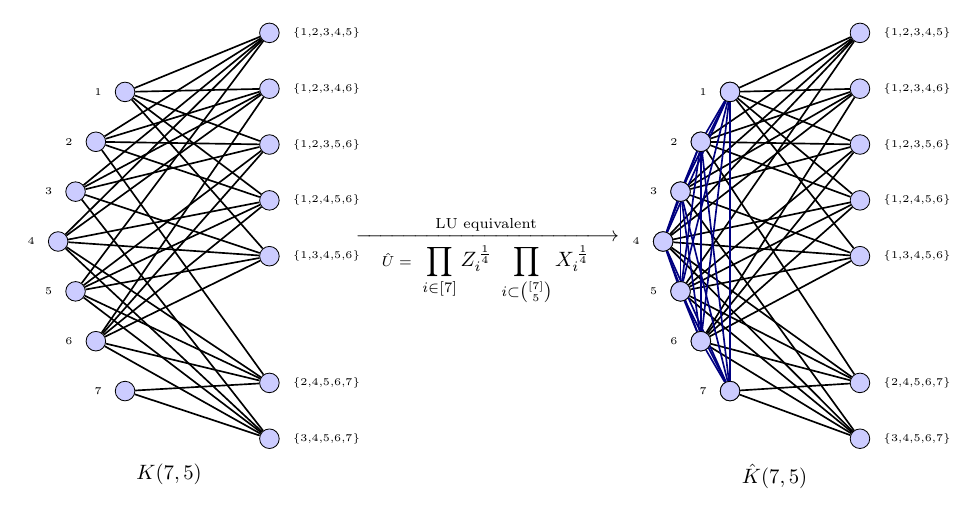}
    \caption{To the left, the bipartite Kneser graph \( K(7,5) \) is spanned between the vertex sets \( V_1 := [7] \) and \( V_2 := \binom{[7]}{5} \), i.e., all 5-element subsets of \([7]\). The graph has a total of 28 vertices. To the right, the graph \(\hat{K}(7,5)\) is shown, which is obtained from \(K(7,5)\) by additionally connecting all vertices in \( V_1 = [7] \). The corresponding graph states are LU-equivalent but not LC-equivalent. Their equivalence is established via the operator \(\hat{U}\), which aplies \( Z^{1/4}_i \) to all vertices \( i \in V_1 \) and \( X^{1/4}_i \) to all vertices \( i \in V_2 \). This example is discussed in Ref.~\cite{LUnoLC}.
}
    \label{fig:K_graph}
\end{figure*}

\section{Reduction to LP problem in modular arithmetic}
\label{sec:StepII}

\cref{Th:reduction}  reduces the problem of finding LU transformation between arbitrary graph states to the problem of finding an LU transformation between specific graph states and of a specific form (\ref{eq:ZX_main}). 
Fortunately, for those graphs we can rely on formulas derived in Ref.~\cite{G_hne_2014,LUnoLC} to track the changes in the graphs by LU-operators of the form (\ref{eq:ZX_main}). 
For this, we shall use the language of weighted hypergraph states. 
We will present LU-operators of the form (\ref{eq:ZX_main}) might as a graphical rules on hypergraph states.

Firstly, notice that a single-qubit operator $Z^\alpha_i$ acting on $i$th register transforms arbitrary  weighted hypergraph state in a simple way. 
Indeed, consider a weighted hypergraph $H_w=(V,f)$ and arbitrary vertex $i\in V$. 
Then by acting with $Z^\alpha_i$ on a weighted hypergraph state $\ket{H_w}$, we obtain another weighted hypergraph state
\begin{equation}
\label{eq:Z_action}
Z^\alpha_i\ket{H_w} 
=
\ket{C_{\{i\},\alpha }\,H} 
\end{equation}
where $C_{\{i\},\alpha }\,H$ is a weighted hypergraph obtained from $H_w$ by adding the weight $\alpha$ to hyperedge $\{i\}$. 

Secondly, we present how a single-qubit operator $X^\alpha_i$ acting on $i$th register transforms arbitrary hypergraph state. 
For our purpose it is enough to present such action in the restricted setting, for general formula, we refer to \cite{LUnoLC}. 
Consider a hypergraph $H=(V,E)$ and arbitrary vertex $i\in V$, and assume that all hyperedges adjacent to $i$ are of size two, i.e. $|e|=2$ for all $e\in E$ such that $i\in e$. 
Moreover, denote by $\delta_i:=\{v\in V: \{i,v\}\in E\}$ the neighborhood of $i$. 
Then by acting with $X^\alpha_i$ on a hypergraph state $\ket{H}$, we obtain a weighted hypergraph state
\begin{equation}
\label{eq:X_action}
X^\alpha_i\ket{H} 
=
\ket{C_{\delta_i,\alpha }\,H}
\end{equation}
where $C_{\delta_i,\alpha }\,H$ is a weighted hypergraph obtained from $H$ by adding the weights $(-2)^{|e|-1}\alpha$ to all hyperedges $e\subset \delta_i$. 

Notice that unlike in \cref{eq:Z_action}, \cref{eq:X_action} assumes action on a hypergraph state, and does not extend to weighted hypergraphs. Nevertheless, as it was observed in Ref.~\cite{LUnoLC}, if vertices $i,j$ are disconnected, the action $X_i^{\alpha_i}$ and $X_j^{\alpha_j}$ commute. In particular, consider graph $G=(V,E)$ and a vertex set $[k,n]:=\{k+1,\ldots,n\}\subset V$. Assume that all vertices in $[k,n]$ are mutually disconnected. 
Then the joint action of $(X_i^{\alpha_i})_{i\in [k,n]}$ operators on graph state $\ket{G}$ is the following
\begin{equation}
\label{eq:X_mluti_action}
\prod_{i=k+1}^n
X_i^{\alpha_i}
\ket{G} = \ket{\prod_{i=k+1}^n C_{\delta_i,\alpha_i}\,H}
\end{equation}
where $\delta_i$ is a neighborhood of $i$, and we use notation (\ref{eq:X_action}). 
Furthermore, all such operators $C_{\delta_i,\alpha_i}$ pairwise commute. 

After analyzing how the operators $X_i^{\alpha_i}$ and $Z_i^{\alpha_i}$ act on graph and weighted hypergraph states, we can reformulate the statement of \cref{Th:reduction} in terms of hypergraphs and, ultimately, as a linear programming (LP) problem in modular arithmetic. In short, one needs to find values $\alpha_i$ for $i\in [k,n]$ for operators $X_i^{\alpha_i}$ acting on $\hat{G}$

\begin{theorem}
\label{prop:solution_2}
Consider graph states $\hat{G}$ and $\hat{G'}$ corresponding to the adjacency matrices (\ref{auxiii_main_main}). There exists LU transformation $\hat{U}$ of the form (\ref{eq:ZX_main}) between $|\hat{G}\rangle$ and $|\hat{G'}\rangle$ if and only if there exists solution to the following system of $2^k-k$ linear equations with $n-k$ variables $\vec{\alpha}:=(\alpha_{k+1},\ldots,\alpha_n)$ in modular arithmetic (mod $2$):
\begin{equation}
\label{eq:11}
    A\cdot\Vec{\alpha} =\Vec{b} \quad (\text{mod}\, 2),
\end{equation}
where $A$ is $(2^{k}-k)\times (n-k)$ matrix with rows indexed by all subsets $e\subset[k]$ such that $|e|\neq 1$ and columns indexed by $i\in [k,n]$ defined as $A_{e,i}=2^{|e|-1}$ if $e\subset \delta_i$ where $ \delta_i$ is a neigberhood of $i$ in graph $G$, otherwise $A_{e,i}=0$; and vector $\Vec{b}=(b_e)$ has rows indexed by all subsets $e\subset[k]$ such that $|e|\neq 1$, with the entry $b_{e}=1$ if $|e|=2$ and $e\in G \xor e\in G'$, and $b_{e}=0$ otherwise.
%
\end{theorem}

\begin{proof}
Assume now, that graph states $|\hat{G}\rangle$ and $|\hat{G'}\rangle$ are LU equivalent by the operator $\hat{U}$ of the form (\ref{eq:ZX_main}). Using (\ref{eq:X_mluti_action}) followed by (\ref{eq:Z_action}), we have:
\begin{align}
\ket{\hat{G}'}
= \ket{\prod_{i=1}^k  C_{\{i\},\alpha_i} \cdot \prod_{i=k+1}^n  C_{\delta_i,\alpha_i}\,\hat{G}},
\label{eq:kkk7}
\end{align}
where the weighted hypergraph on the right side is obtained from graph $\hat{G}$ by adding the weights $(-2)^{|e|-1}\alpha_i$ to all hyperedges $e\subset \delta_i$ for all $i=k+1,\ldots,n$ and by adding $+\alpha_i$ weights to all the single-vertex hyperedges $i\in [k]$. Since weighted hypergraph states are equal if and only if they correspond to the same hypergraphs, the hypergraphs on both sides of (\ref{eq:kkk7}) must be identical. Notice that weights might be different only on hyperedges $e\subset[k]$ which label equations in the system (\ref{eq:11}). We further investigate weights by those hyperedges in three separate cases: $|e|>2,|e|=2,|e|=1$.
As $\hat{G}'$ is a graph, in the hypergraph on the right side the weights for all hyperedges $e$ of the size $|e|> 2$ on the right-hand side must vanish modulo two, leading to (\ref{eq:11}) for all $e:|e|>2$. Similarly, analyzing the total weights for all hyperedges $e$ of the size $|e|= 2$ on both sides of (\ref{eq:kkk7}) results in remaining equations in (\ref{eq:11}). For each $e$ with $|e|=1$, an additional equation of type (\ref{eq:11}) indexed by $e \in [k]$ exists. However, since each equation involves a single independent variable $\alpha_e$, both the equations and variables $\alpha_e$ for $e \in [k]$ can be omitted.
\end{proof}


\section{Solution of LP problem}
\label{AppB}

In the two previous sections, we effectively reduced the problem of verifying LU-equivalence between a given pair of graph states to solving a linear system of equations in modular arithmetic (see \cref{prop:solution_2}). Note, however, that system (\ref{eq:11}) can have exponentially many equations in the number of nodes $n$. In this section, we discuss the complexity of solving such a system. In particular, we demonstrate that the problem can be restricted to a discrete set of solutions, see \cref{corollary:reduction_Further3}. Consequently, we discuss strategies to circumvent the challenge of dealing with a potentially exponential number of equations in $n$, as detailed in \cref{lemma:reduction_LC_r}. 
%

Subsequently, we discuss examples of known pairs of states that are LC-equivalent but not LU-equivalent and comment on how our algorithm performs for such cases. Based on those findings, we formulate conjectures regarding the feasibility of determining LU-equivalence between arbitrary graph states, see \cref{con:2}. 

\begin{proposition}
\label{proposition:reduction_Further2}
Consider two graphs $\hat{G}$ and $\hat{G'}$ with adjacency matrices of the form (\ref{auxiii_main_main}), and denote by $\hat{G}_\Delta$ and $\hat{G'}_\Delta$ graphs obtained from $\hat{G}$ and $\hat{G'}$ by deleting vertices from the set $[k,n]$ that have the same neighborhood, keeping only one representative vertex for each set of vertices with identical neighborhoods. The corresponding graph states $|\hat{G}\rangle$ and $|\hat{G'}\rangle$ are LU equivalent by the operator of the form (\ref{eq:ZX_main}) if and only if graph states $|\hat{G}_\Delta\rangle$ and $|\hat{G'}_\Delta\rangle$ are LU equivalent by the operator of the form (\ref{eq:ZX_main}).
\end{proposition}

\begin{proof}
Notice that if vertices $i,j \in [k,n]$ share the same neighborhood in graph $\hat{G}$, they also share the same neighborhood in $\hat{G'}$. Moreover, variables $\alpha_i$ and $\alpha_j$ always appear in the same equations in \cref{eq:11}, meaning that the $i$-th and $j$-th columns of matrix $A$ are the same. Therefore, deleting one of these variables, say $\alpha_i$, and the corresponding $i$-th row in matrix $A$ does not change the solvability of the system (\ref{eq:11}). 
In that way, we delete vertices from $[k,n]$ that have the same neighborhood without changing solvability of the system.
\end{proof}

As a consequence of \cref{proposition:reduction_Further2}, the solvability of the system (\ref{eq:11}) can be investigated over a discrete set of potential solutions. Furthermore, the number of non-trivial equations can be reduced, namely we have the following result.

\begin{corollary}
\label{lemma:reduction_Further3}
The linear system of equations (\ref{eq:11}) has a solution if and only if there exists a solution $\vec{\alpha} := (\alpha_{k+1}, \ldots, \alpha_n)$ for $\alpha_i\in V_{|\delta_i|-1}$ where
\begin{equation}
    \label{eq:V_r}
    V_r:=\Big\{ \frac{m}{2^{r-1}} \,: \,m=0,1,\ldots,2^r-1\Big\}
\end{equation}
is a discrete set of $2^k$ elements, and $\delta_i$ denotes the neighborhood of vertex $i$ in either of the graphs $\hat{G}$ or $\hat{G'}$. 

Moreover, in (\ref{eq:11}) there are at most $ 2^{\Delta_G} (n-k)$ non-trivial equations, where
\begin{equation}
    \Delta_G:=
    \max_{i\in[k,n]} \; |\delta_i |,
\end{equation} 
is the maximal vertex-degree among vertices $i\in[k,n]$.
\end{corollary}

\begin{proof}
Consider two graphs $\hat{G}$ and $\hat{G'}$ with adjacency matrices of the form (\ref{auxiii_main_main}) corresponding to the linear system (\ref{eq:11}). Assume that all vertices $i \in [k,n]$ have pairwise distinct neighborhoods, i.e., $\delta_i \neq \delta_j$ for $i \neq j \in [k,n]$. Then, for any $i \in [k,n]$, the equation in (\ref{eq:11}) corresponding to $e=\delta_i$ takes the form $2^{|\delta_i|-1} \alpha_i = 0 \,\, (\text{mod} \, 2)$. Furthermore, all equation in (\ref{eq:11}) are in modular arithmetic $ (\text{mod} \, 2)$, hence we can take $\alpha_i \in  V_{|\delta_i|-1}$, which proves the first statement in this case. 

Now assume that some vertices $i_1, \ldots, i_s \in [k,n]$ share the same neighborhood, i.e., $\delta_{i_1} = \cdots = \delta_{i_s}$. As shown in \cref{proposition:reduction_Further2}, the system (\ref{eq:11}) has a solution if and only if there exists a solution with $\alpha_{i_2} = \cdots = \alpha_{i_s}$. Under this assumption, the equation in (\ref{eq:11}) corresponding to $\delta_{i_1}$ takes the form $2^{|\delta_{i_1}|-1} \alpha_{i_1} = 0 \,\, (\text{mod} \, 2)$, completing the proof of the first statement.

Notice that the only equations in (\ref{eq:11}) which are non-trivially satisfies correspond to $e$ such that $e\subset\delta_i$ for some $i\in[k,n]$. There are at most $ 2^{\Delta_G} (n-k)$ such equations. 
\end{proof}

Notice that, as $\Delta_G,\,k <n$ the number of non-trivial equations in this system is exponential in number of vertices $n$. One approach to addressing this problem is to restrict our consideration to solutions within a constrained subspace. As we shall see, this also reduces the number of non-trivial equations in the system.

\begin{corollary}
\label{lemma:reduction_LC_r}
Fix parameter $r\in[n]$ and consider the linear system of equations (\ref{eq:11}) with solutions $\vec{\alpha} := (\alpha_{k+1}, \ldots, \alpha_n)$ restricted to $\alpha_i\in V_{r}$. Such system has at most $(n-k)\sum_{t=2}^r\binom{\Delta_G}{t}$ non-trivial equations.

Moreover, there is the following upper bound for number of non-trivial equations $r(n-k)\binom{\Delta_G}{r}\leq n^{r+1}$.
\end{corollary}

\begin{proof}
Notice that the only equations in (\ref{eq:11}) that are nontriviallyatisfies correspond to $e$ such that $e\subset\delta_i$ for some $i\in[k,n]$, and $1<|e|<r$. For each $i\in[k,n]$ there are exactly $\sum_{t=2}^r\binom{\Delta_G}{t}$ such sets $e$, which gives the estimate. The aforementioned upper bound is a simple calculation. 
\end{proof}

Notice that $V_1\subset V_2\subset\cdots\subset V_{\Delta_G}$, and thus, solving the system (\ref{eq:11}) within the solution space $\alpha_i\in V_r$ for increasing $r$ provides a natural hierarchy for finding the LU operator that transforms two graph states. Furthermore, the number of non-trivial equations can be bounded by $n^{r+1}$ for $r<<n$ and by $2^n$ for larger $r$. Moreover, this hierarchy is complete. In fact, as a consequence of \cref{lemma:reduction_Further3}, the linear system of equations (\ref{eq:11}) has a solution if and only if a solution exists for $\alpha_i\in V_{\Delta_G}$. 

In fact, this hierarchy coincides with the hierarchy of equivalences of graph states discussed in Ref.~\cite{claudet2024localequivalencestabilizerstates}, where the authors defined the class of $r$-local complementations as:
\begin{equation}
    \LC_r:= \langle H, Z^{1/2r}\rangle .
\end{equation}
It is easy to observe that solving the system (\ref{eq:11}) within the solution space $\alpha_i \in V_r$ is equivalent to finding an $ \text{LC}_r$-operator as discussed in Ref~\cite{claudet2024localequivalencestabilizerstates}. 

For example, verifying $\LC_1$-equivalence is in fact $\LC$-equivalence. Therefore, verifying LC-equivalence corresponds to solving the linear system of equations (\ref{eq:11}) with solutions restricted to $\alpha_i\in V_{1}=\{0,1\}$. 


Interestingly, Ref~\cite{claudet2024localequivalencestabilizerstates} shows that $r$-local complementations form a strict hierarchy. Indeed, the authors demonstrated that for
\begin{align}
\label{eq:t_k}
 k \,=\,& 2r + 2^{\lfloor \log_2(r) \rfloor + 1} - 1, \\
\nonumber   
      t \,=\,& 2r + 1, 
\end{align}
the pair of states \( K(k,t) \) and \( \hat{K}(k,t) \), see \cref{ex:K_graph} for definition, is $\text{LC}_{r}$-equivalent but not $\text{LC}_{r-1}$-equivalent.  In other words, using our algorithm, system (\ref{eq:11}) can be solved for $\alpha_i \in V_{r}$, but there is no solution for $\alpha_i \in V_{r-1}$. In particular, choosing $r=2$, we recover pair of states \( K(7,5) \) and \( \hat{K}(7,5) \) which are $\text{LC}_{2}$-equivalent but not $\text{LC}$-equivalent (i.e. $\text{LC}_{1}$-equivalent), see \cref{fig:K_graph}. Notably the number of vertices in such graphs is exponential in $r$. Indeed, we can lower bound the number of vertices in the graphs $K(k,t) $ and $\hat{K}(k,t)$ for the case defined in (\ref{eq:t_k}), by applying Stirling’s approximation, we find that $|K(k,t) |=|\hat{K}(k,t)|>e^r$. This leads to the following.  

\begin{conjecture}
\label{con:2}
    There exists a function $F\in \mathcal{O}(\text{log}\,n)$ such that (\ref{eq:11}) is solvable if and only if it is solvable in  space $V_{F(n)}$. As a result, our algorithm can decide if the initial graph states are LU-equivalent by solving $n^{\mathcal{O}(\text{log} \,n)}$ equations in (\ref{eq:11}).

    Equivalently, this can be stated as $\LU=\LC_{F(n)}$ for all stabilizer states on $n$ qubits, meaning all LU-equivalent stabilizer states on $n$ qubit are $\LC_{F(n)} $-equivalent.  
\end{conjecture}

\noindent
Notice the second part of \cref{con:2} follows from \cref{lemma:reduction_LC_r}, where we have shown that solving (\ref{eq:11}) in the restricted space $V_r$ limits the number of non-trivial equations to $n^{F(n)+1}$.

We conclude this section by noting that the linear system of equations (\ref{eq:11}) with solutions in $\alpha_i\in V_r$ can be readily transformed into a linear system over $\mathbb{F}_2$. For further details, we refer to \cref{Appendix_X}.

\section{Algorithm outline}
\label{Algorithm_outline}

In this section, we combine results from \cref{sec:Necessary conditions,sec:StepI,sec:StepII} to reduce the problem of finding an LU transformation between two graph states to a linear programming (LP) problem in modular arithmetic.

We begin by examining necessary conditions for the existence of such an LU transformation in terms of minimal local sets. Specifically, we find an MLS cover $\mathcal{M}$ for $G$, check that all sets from $\mathcal{M}$ are minimal in $G'$, and verify that their types match, as required by \cref{prop:MLS_cover_same}.

If $\mathcal{M}$ is an MLS cover for both $G$ and $G'$, we construct the pairs of functions $F_G^{(\ell)}, F_{G'}^{(\ell)}: [n] \to {X,Y,Z}$ indexed by $\ell \in \mathcal{L}$, as described in Algorithm 1. According to \cref{prop_pairs_of_func}, any LU operator $U = U_1 \otimes \dots \otimes U_n$ that satisfies $U\ket{G} = \ket{G'}$ must also satisfy $F_{G'}^\ell(i) = U_i \,F_G^\ell(i)\, U_i^\dagger$ for at least one $\ell \in \mathcal{L}$. 

For each $\ell \in \mathcal{L}$, we assume that $U\ket{G} = \ket{G'}$ satisfies the condition above and construct the pairs of graphs $\hat{G}, \hat{G'}$, where $|\hat{G}\rangle$ is LC-equivalent to $\ket{G}$ and $|\hat{G'}\rangle$ to $\ket{G'}$, as outlined in Algorithm 2. According to \cref{Th:reduction}, any LU equivalence between $|\hat{G}\rangle$ and $|\hat{G'}\rangle$ has the form (\ref{eq:ZX_main}).

Next, following \cref{prop:solution_2}, we transform this problem into a linear system of equations in modular arithmetic, as described in Algorithm 3. We solve this system using LP methods. If a solution exists, it confirms LU equivalence between $\ket{G}$ and $\ket{G'}$. If no solution is found, we continue with other values of $\ell \in \mathcal{L}$. If no solution is found for any $\ell$, the states $\ket{G}$ and $\ket{G'}$ are not LU-equivalent.

\cref{fig:algo_outline} provides a schematic description of the algorithm. 
Moreover, \cref{fig:algo_subroutines,fig:algo_main} present the pseudocode for the algorithm, where the three main reduction steps—Algorithms 1,2, and 3 are used as subroutines.

\subsection{Computational complexity}
\label{sec:comput_complex}

The exact computational complexity of our algorithm remains unknown, particularly with respect to whether it is borderline efficient (quasi-polynomial in the number of nodes, $n$) or infeasible (exponential in $n$). The MLS cover can be computed in $\mathcal{O}(n^4)$ time as it was shown in Ref~\cite[Theorem 1]{claudet2024covering}. For a given MLS cover, Algorithm 1 on \cref{fig:algo_subroutines} computes a class of pairs of functions $F_G^{\ell }, F_{G'}^{\ell }$ indexed by ${\ell \in \mathcal{L}}$. While the computational procedure runs in $\mathcal{O}(n^3)$ time, the size of the outcome class $|\mathcal{L}|$ might be large, in principle exponential in $n$. Note that the rest of the procedure is performed independently for all $\ell\in \mathcal{L}$.

In fact, the size of the class $\mathcal{L}$ can be upper bounded by $|\mathcal{L}|\leq 3^t$ where $t$ denotes the number of connected components in the intersection graph of $\mathcal{M}$ consisting entirely of MLSs of Type II, see \cref{prop_pairs_of_func}. Through random testing of pairs of states up to $n = 12$, we verified that this number satisfies $t=0,1$ and related class $\mathcal{L}$ is of size $|\mathcal{L}|=1,3$. We strongly suspect that $|\mathcal{L}|=o(1)$, but we are currently unable to rigorously prove this in general. Establishing a formal bound on $|\mathcal{L}|$ remains an open question. 

For given $\ell\in \mathcal{L}$, the reduction presented in Algorithm 2 and further Algorithm 3 yields the linear system of equations (\ref{eq:11}). These procedures involve Gaussian elimination, which can be performed in $\mathcal{O}(n^3)$ time. As we presented in \cref{AppB}, this system can then be transformed into a linear system over $\mathbb{F}_2$ with number of variables bounded by $n^2$ and a number of equations bounded by $2^n$. Despite the number of equations being exponential in $n$, the exact complexity of solving such a system remains unclear. Denoting this complexity as 
$C(n)$, the overall complexity of our algorithm is given by
\begin{equation}
    \mathcal{O} 
    \Big(n^4+
    |\mathcal{L}| (n^3 +C(n))\Big).
\end{equation}
We leave the exact estimation of $C(n)$ and $|\mathcal{L}|$ as open problems. 


Assuming the bound $|\mathcal{L}|=o(1)$ and \cref{con:2}, the complexity of our algorithm becomes
\begin{equation}
n^{\mathcal{O}(\text{log}\, n)}
\end{equation}
which is quasi-polynomial in the number of nodes $n$. As noted in \cref{AppB}, the existing examples of LU-equivalent graph states suggest that our algorithm cannot decide the the LU equivalence in polynomial time. We suspect that an efficient algorithm for this problem, i.e. one running in time polynomial in $n$, is unlikely to exist.



\begin{figure}
\centering
\includegraphics[width = 0.9\linewidth]{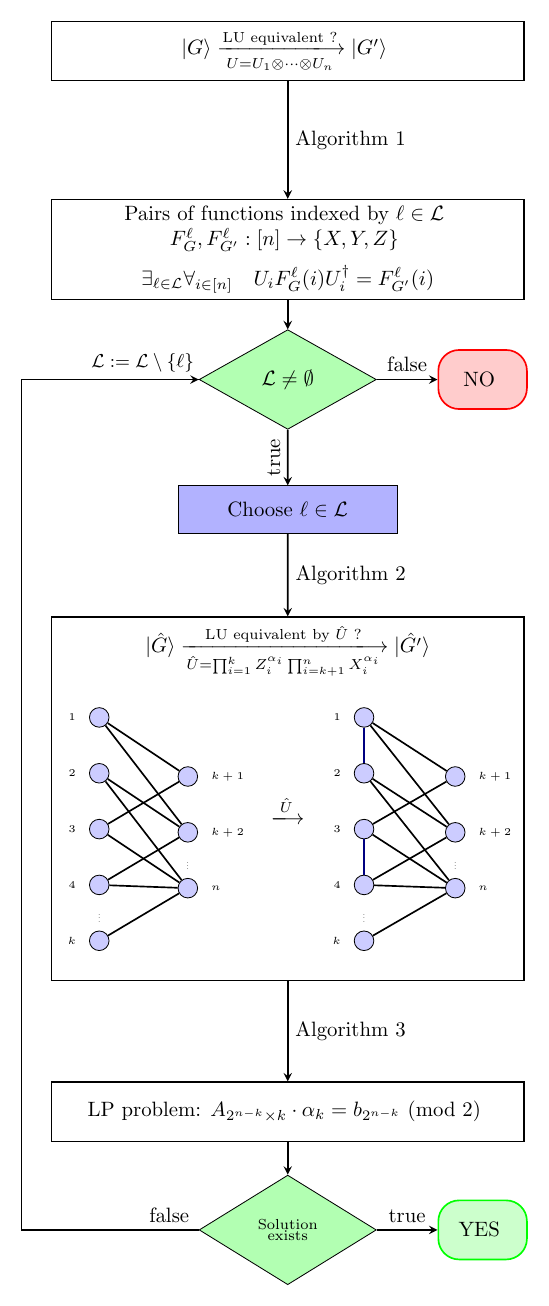}
\caption{A schematic description of the algorithm to verify the existence of an LU-equivalence consists of three main reduction steps: Algorithms 1,2, and 3 presented on \cref{fig:algo_subroutines}.  
}
\label{fig:algo_outline}
\end{figure}

\section{Analysis of graph states up to $n=11$ qubits} 
\label{sec:numerics}

Graph states have been classified up to $n = 12$ qubits \cite{Danielsen_2008,cabelloEntanglementEightqubitGraph2009}, and there is a database containing all unlabelled LC-orbits up to $n = 12$ vertices, see Ref~\cite{data_graph_states}. We make use of this database in our analysis. Consider any pair of unlabelled LC-orbits of graphs on $n$ vertices. For each pair, we can select their representatives $G$ and $G'$ with arbitrarily chosen labelling of the vertices. For every permutation $\sigma \in S_n$ of the vertices of $G'$, we first check if $G$ and $\sigma(G')$ have the same $\CutRank$ function, which is a necessary condition for LU equivalence between $G$ and $\sigma(G')$. These $\CutRank$ functions are calculated by making use of the $LU$-invariants of graph states as presented in Ref~\cite{Lina2024}.

For $n < 9$, no pairs of graphs $G$ and $\sigma(G')$ from different unlabelled LC-orbits share the same $\CutRank$ functions. For $n = 9, 10, 11$, there are $1, 3, 15$ pairs of LC-orbits, respectively, for which permutations $\sigma(G')$ exist such that $\CutRank(G) = \CutRank(\sigma(G'))$. We further examine these cases and run our algorithm to verify LU equivalence for the selected pairs $G$ and $\sigma(G')$ that share the same $\CutRank$ function. In almost all cases, the algorithm terminates in the first subroutine (Algorithm 1), verifying that the necessary conditions for LU equivalence are not satisfied.

Only for one pair of $n = 11$ qubit LC-orbits, the algorithm returned four compatible functions $F_{G}^\ell,F_{\sigma(G')}^\ell$, $\ell \in [4]$. This corresponds to the (unlabelled) orbits \textit{10010} and \textit{14666} in the catalog from Ref~\cite{data_graph_states}, but note that their $\CutRank$ functions only coincide for a non-trivial permutation of the representative of \textit{14666} as presented in \cite{data_graph_states}.

There are exactly four permutations $\sigma \in S_n$ for which these four compatible functions $F_{G}^\ell,F_{\sigma(G')}^\ell$ exist, one of which is shown in \cref{fig:example_11qubit}. For all four permutations and their corresponding families $F_{G}^\ell,F_{\sigma(G')}^\ell$, our algorithm proceeds to the second subroutine (Algorithm 2), where it verifies that these classes are not LU equivalent. 
These findings are summarized in \cref{tab:table1}. 

As a result, we conclude that up to $n = 11$ qubits, no graphs from different unlabelled LC-orbits are LU equivalent. Thus, the number of unlabelled LU- and LC-orbits up to $n = 11$ is exactly the same. This was previously known up to $n = 8$ \cite{cabelloEntanglementEightqubitGraph2009}, and recently improved to $n = 10$ \cite{Lina2024}.

\begin{figure}
\centering
\includegraphics[width = 0.8\linewidth]{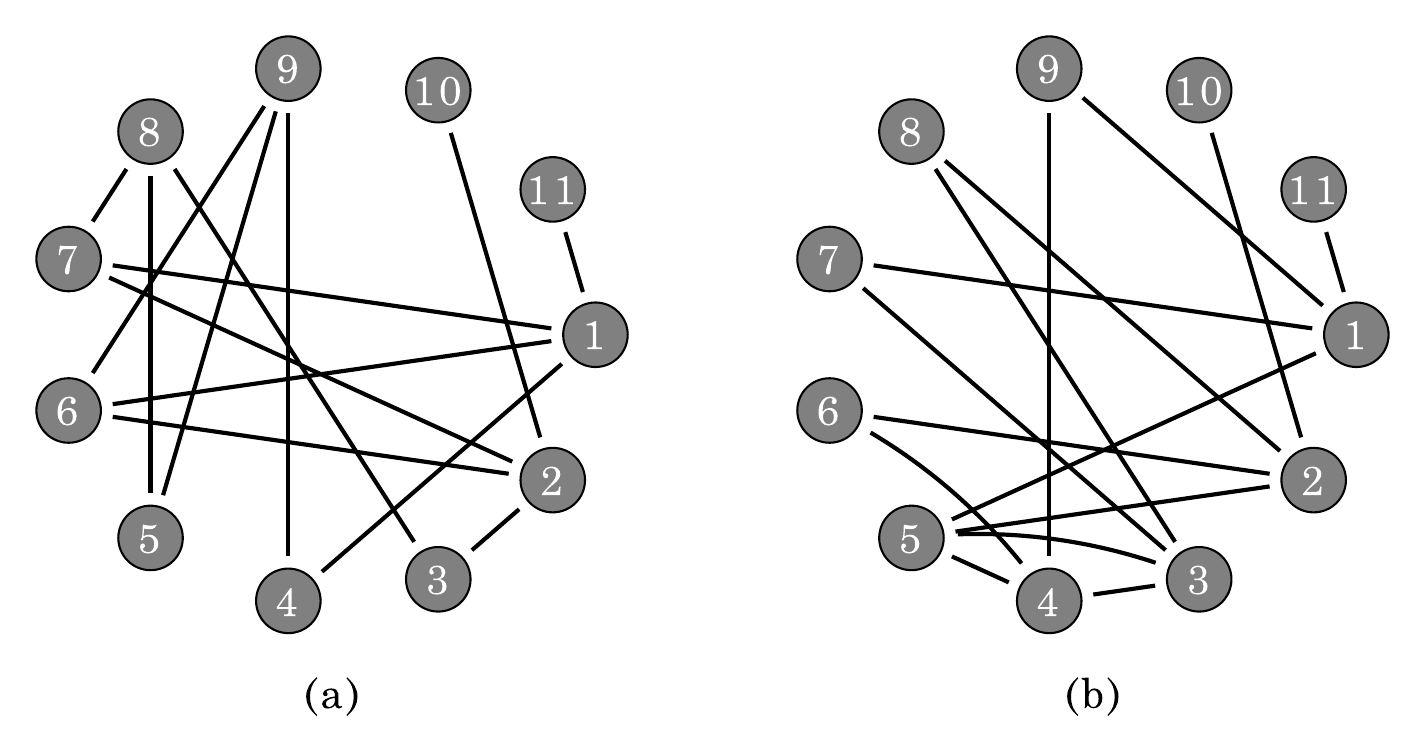}
\caption{Two eleven-qubit graphs that belong to different LC-orbits, have the same $\CutRank$ function, and hence the same entanglement properties. We verified that they are not LU equivalent by running our algorithm.}
\label{fig:example_11qubit}
\end{figure}

\begin{table}[htbp]
\centering
\begin{tabular}{c|c|c|c}
 $n$  	& \#  LC-orbits 	& \#  cut-rank 	& \#  compatible family 
 \\ 
&  (unlabelled)	& function pairs 	&  of function $F$ pairs
 \\ 
	 \hline
2   	& 1	 & -	 & -	    \\ 
3   	& 1	 & -	 & -	    \\ 
4   	& 2	 & -	 & -	    \\ 
5   	& 4	 & -	 & -	    \\ 
6   	& 11	 & -	 & -	    \\ 
7   	& 26	 & -	 & -	    \\  
8   	& 101	 & -	 & -	    \\  
9   	& 404	 & 	1 & -	    \\  
10   	& 3\ 132	 & 3	 & -	    \\  
11  	&  40\ 457	 & 15	 & 1    \\ 
\end{tabular}%
\caption{
The first column gives the number of vertices of the graph. The second column provides the number of LC-orbits \cite{data_graph_states} of unlabelled graphs on $n$ vertices. The third column shows the number of pairs of LC-orbits for which there exists a labelling of vertices with the same cut-rank function. The fourth column gives the number of pairs of LC-orbits for which there exists a labelling of vertices with compatible family of functions $F_{G}^\ell,F_{\sigma(G')}^\ell$.
}%
\label{tab:table1}
\end{table}

\section{Concluding remarks}
\label{Discussion}

In this paper, we present a novel algorithm for verifying local unitary (LU) equivalence between stabilizer states. 
Since any stabilizer state is LC equivalent to a graph state \cite{vandennestGraphicalDescriptionAction2004}, we can restrict our problem to determining LU equivalence between graph states. 

Our algorithm takes two arbitrary graph states as input and outputs a linear system of equations in modular arithmetic, whose solvability is equivalent to the equivalence of the initial graph states. 
\cref{fig:algo_outline} graphically presents the algorithm in which the three main steps of the algorithm are based on \cref{prop:MLS_cover_same}, \cref{prop:solution_2}, and \cref{prop:solution_2} respectively.


Using existing libraries, we verify that for up to $n = 11$, the number of LU and LC orbits of stabilizer states is identical. This finding extends previous results on smaller number of qubits~\cite{cabelloEntanglementEightqubitGraph2009,Lina2024}.


Despite its practical effectiveness, the exact computational complexity of our algorithm remains unknown, particularly whether it is borderline efficient (quasi-polynomial in the number of nodes $n$). The open challenges include rigorously bounding the size of the class of compatible functions $\mathcal{L}$ and determining the precise complexity $C(n)$ of solving the linear systems of equations output by our algorithm (\ref{eq:11}), see the discussion in \cref{sec:comput_complex}. 


\section*{Acknowledgement}
We express our heartfelt gratitude to Otfried G\"uhne, Tim Coopmans, and Sebastiaan Brand for their invaluable discussions and insightful contributions to this work. 

Moreover, the authors wish to thank Nathan Claudet and Simon Perdrix, who were working on closely related problems, for their valuable correspondence and insightful exchanges during the preparation of the second version of this manuscript. 

This research was supported by an NWO Vidi grant (Project No. VI.Vidi.192.109),
 the European Union (ERC, ASC-Q, 101040624), 
the Deutsche Forschungsgemeinschaft  (DFG, German Research Foundation, project numbers 447948357 and 440958198), the Sino-German Center for Research Promotion (Project M-0294), the German Ministry of Education and Research (Project QuKuK, BMBF Grant No.\ 16KIS1618K), the European Union via the Quantum Internet Alliance project, and the Stiftung der Deutschen Wirtschaft.

Views and opinions expressed are however those of the authors only and do not necessarily reflect those of the European Union or the European Research Council. Neither the European Union nor the granting authority can be held responsible for them.

It should be noted that a closely related study appeared on arXiv at approximately the same time as our work \cite{claudet2024localequivalencestabilizerstates}. We were unaware of this research during the preparation of our manuscript, and the results presented here were obtained independently. In the current version of our manuscript, particularly in \cref{AppB}, we have added comments to clarify how our findings relate to those presented in Ref.~\cite{claudet2024localequivalencestabilizerstates}.

\section*{Appendices}

\cref{app:MLS} contains the proof of \cref{prop:necessary_conditions_MLS}. 
\cref{appendix:D} provides the proof of \cref{prop_pairs_of_func}, which establishes the correctness of Algorithm 1. 
In \cref{appendix:0}, we present several observations regarding transition matrices. In \cref{appendix:A}, we present results concerning check matrices of graph and stabilizer states. These results are later applied in \cref{appendix:B}, for the proof of \cref{Th:reduction}. Finally, \cref{Appendix_X} contains some additional remarks on \cref{AppB}.

\appendix
\section{Proof of \cref{prop:necessary_conditions_MLS}}
\label{app:MLS}

\begin{proof}[Proof of \cref{prop:necessary_conditions_MLS}]
Suppose that $U=U_1\otimes\cdots\otimes U_n$ is an LU transformation transforming $\ket{G}$ into $\ket{G'}$. 
In particular, it transforms reduced states $\rho_M (G)$ into $\rho_M (G')$. 
Hence, by (\ref{eq:marginalstab_def}), we have
\begin{align}
\label{eq:type_both}
\frac{1}{2^{|M|}}\sum_{S\in \St_{\ket{G'}}^M} S_{\downarrow M} &=\rho_M (G')=
U\rho_M (G) U^\dagger =\\&=
\frac{1}{2^{|M|}}
U \big( \sum_{S\in \St_{\ket{G}}^M} S_{\downarrow M} \big) U^\dagger
\nonumber
\end{align}
which leads to
\begin{equation}
\label{eq:type_IandII}
\sum_{\substack{s'\in \St_{\ket{G'}}^M\\ s'\neq \id^{\otimes n}}} s' =
U \big( \sum_{\substack{s\in \St_{\ket{G}}^M\\ S\neq \id^{\otimes n}}} s \big) U^\dagger
\end{equation}
For $M$ being a MLS we shall consider the two cases when $M$ is of Type I or II separately.

Notice that if $M$ is a MLS of Type I, the sums in \cref{eq:type_IandII} contain only one element, and hence the statement follows immediately. 

Suppose $M$ is a MLS of Type II, then sums in \cref{eq:type_IandII} contain three elements each. 
Denote by $s_1,s_2,s_3\in \St_{\ket{G}}^M$ and  $s_1',s_2',s_3'\in \St_{\ket{G'}}^M$ all non-trivial elements in $\St_{\ket{G}}^M$ and $\St_{\ket{G'}}^M$ respectively. 
Without loss of generality, suppose $M=[k]$ is a set of the first $k$ indices, and express Pauli strings
\[
s_i=s_{i1} \otimes \cdots\otimes s_{ik} \otimes \id^{\otimes n-k}
\]
for $i=1,2,3$ and $s_{i1}, \dots s_{ik} \in \{X, Y, Z\}$, similarly for $s_i'$. 
According to \cite{Lina2024}, $k\geq 4$ 
and for each $j\in [k]$, the Pauli matrices $s_{1j},s_{2j},s_{3j}$ are all different, hence they span the entire set $\{X,Y,Z\}$. Since single qubit operations in general can be expressed as a linear combination of $\{\id, X,Y,Z\}$ and since unitary transformations do not change the trace, we have for each $j\in [k]$ and $i \in \{1,2,3\}$: 
\begin{align}
\label{eq:auxy0.5}
U_j s_{ij} U_j^\dagger =\sum_{j'=1}^3 c^{(i)}_{jj'}s'_{ij'}?
\end{align}
for some coefficients $c^{(i)}_{jj'}$ that satisfies 
\begin{equation}
\label{eq:auxy0}
|c_{j1}^{(i)}|^2+|c_{j2}^{(i)}|^2+|c_{j3}^{(i)}|^2=1,
\end{equation}
as the matrix $(c^{(i)}_{jj'})_{j,j'=1,2,3}$ is unitary \footnote{Indeed, for arbitrary Pauli matrix $\sigma$ and unitary matrix $V$, we have $V^\dagger \sigma V= c_X X+c_Y+c_Z Z$ for coefficients $c_X,c_Y,c_Z$ that satisfies $|c_X|^2+|c_Y|^2+|c_Z|^2=1$. Indeed this follows from the fact that $\{\id,X,Y,Z\}$ forms a basis of unitary matrices.}. 
consequently, for each $i\in [k]$, we have 
\begin{equation}
\label{eq:auxy1}
\sum_{j,j'=1}^3 |c^{(i)}_{jj'}|^2=3.
\end{equation}
On the other hand, from \cref{eq:type_IandII}, we have
\[
\sum_{j=1}^3 \Big(\prod_{i\in [k]} c_{jj'}^{(i)}\Big) =1
\]
for arbitrary $j'=1,2,3$. Hence, we have
\begin{equation}
\label{eq:auxy2}
\sum_{j,j'=1}^3 \Big(\prod_{i\in [k]} |c_{jj'}^i|\Big) \geq 3.
\end{equation}
As we shall see, \cref{eq:auxy1,eq:auxy2} together with the fact that $k>2$ shows that we have only trivial solutions. 
Indeed, from (\ref{eq:auxy1}) and inequalities between geometric ant arithmetic averages \footnote{for non-negative real numbers $x_1,\ldots,x_k$, we have $k x_1\cdots x_k\leq x_1+\cdots+x_k$, where the equality holds iff $x_1=\cdots =x_k$.}, we have 
\begin{equation}
\label{eq:auxy3}
\sum_{j,j'=1}^3 \Big(\prod_{i\in [k]} |c_{jj'}^{(i)}|\Big) \leq 3.
\end{equation}
with equality iff for each $j,j'=1,2,3$, we have $|c_{jj'}^{(i)}|=|c_{jj'}^{(i')}|$ for all $i,i'\in [k]$. 
Notice that from (\ref{eq:auxy2}), there is equality in (\ref{eq:auxy3}), hence indeed for each $j,j'=1,2,3$, we have 
\begin{equation}
\label{eq:auxy3.5}
|c_{jj'}^{(i)}|=|c_{jj'}^{(i')}|
\end{equation}
for all $i,i'\in [k]$. 
Therefore (\ref{eq:auxy2}) can be rewritten as
\begin{equation}
\label{eq:auxy4}
\sum_{j,j'=1}^3 |c_{jj'}^{(i)}|^k \geq 3.
\end{equation}
for arbitrary $i\in [k]$. 
Notice that $|c_{jj'}^i| \leq 1$ and hence for $k>2$, $|c_{jj'}^{(i)}|^k \leq |c_{jj'}^i|^2$ with identity iff $|c_{jj'}^{(i)}|=0,1$. 
Therefore, by (\ref{eq:auxy1}), we have that $|c_{jj'}^{(i)}|=0,1$ for all $i,j,j'$. 
Consider $i=1$. From the fact that matrix $(c^{(i)}_{jj'})_{j,j'=1,2,3}$ is unitary, we can conclude that there is a bijection $f_M:\{1,2,3\}\rightarrow \{1,2,3\}$, such that
\begin{equation}
\label{eq:auxy4P}
|c^{(1)}_{jj'}|=
\begin{cases}
  1  & j'=f_M(j) \\
  0 & \text{otherwise}
\end{cases}
\end{equation}
Therefore from (\ref{eq:auxy0.5}) and (\ref{eq:auxy3.5}), for arbitrary $j=1,2,3$, we have $U s_j U^\dagger= \omega_j s_{f_M (j)}$ for some phase $\omega_j$, i.e $|\omega_j |=1$. 
Notice that since there are no phases in (\ref{eq:type_IandII}), $\omega_j=1$ which concludes the statement.
\end{proof}

\section{Proof of \cref{prop_pairs_of_func}}
\label{appendix:D}

Here we present a proof of \cref{prop_pairs_of_func}. We begin by establishing some necessary conditions for two graph states to be LU equivalent.

\begin{lemma}
\label{lemma:overlaps_conditions}
Consider two graphs $G,G'$ sharing the same MLS cover $\mathcal{M}$. Furthermore, consider two sets $M_1,M_2\in \mathcal{M}$ of Type I with non-trivial intersection, i.e. $M_1\cap M_2\neq \emptyset$. Denote by $s_G^{M_1}\in \St_{\ket{G}}^{M_1},s_G^{M_2}\in \St_{\ket{G}}^{M_1}$ and by $s_{G'}^{M_1}\in \St_{\ket{G'}}^{M_1},s_{G'}^{M_2}\in \St_{\ket{G'}}^{M_1}$ the corresponding elements in reduced stabilizer groups. Then, if $\ket{G}$ is LU-equivalent to $\ket{G'}$. we have:
\begin{equation}
\label{eq:supp}
    \supp 
    \big( s_G^{M_1} s_G^{M_2} \big)
    =
    \supp 
    \big( s_{G'}^{M_1} s_{G'}^{M_2} \big).
\end{equation}
\end{lemma}

\begin{proof}
    Suppose that we have $\ket{G'}=U\ket{G}$ for some local unitary operator $U$. By \cref{prop:necessary_conditions_MLS}, we have
    \begin{equation*}
        s_{G'}^{M_1} = U\,s_G^{M_1} \,U^\dagger
        ,
        \quad
        s_{G'}^{M_2}  = U\,s_G^{M_2} \,U^\dagger,
    \end{equation*}
    and hence
    \begin{equation*}
        s_{G'}^{M_1}s_{G'}^{M_2}  
        = U\,(s_G^{M_1}s_G^{M_2}) \,U^\dagger.
    \end{equation*}
    Notice that $s_{G'}^{M_1}s_{G'}^{M_2} $ and $s_G^{M_1}s_G^{M_2}$ are Pauli strings. Since $U$ is a unitary operator, it does not change the support, hence we have  (\ref{eq:supp}), which finishes the proof.
\end{proof}

We proceed with two lemmas that shows how to extend compatible functions to the bigger domain.

\begin{lemma}
\label{lemma:extension}
Consider two graphs $G,G'$ sharing the same MLS cover $\mathcal{M}$. Moreover, let $N\subset [n]$ and suppose there exists LU operator $U=U_1\otimes\cdots\otimes U_n$, transforming $\ket{G}$ to $\ket{G'}$ which satisfies
\begin{equation}
    U_i\,F_G(i)\,U_i^\dagger=F_{G'} (i)
\end{equation}
for some functions $F_G,F_{G'}:N\rightarrow\{X,Y,Z\}$, and all $i\in N$. 

Consider $M\in \mathcal{M}$ which is either of Type I, or satisfies $M\cap N\neq \emptyset$. Then, we can extend functions $F_G,F_{G'}$ to $\widetilde{F_G},\widetilde{F_{G'}} :N\cap M\rightarrow\{X,Y,Z\}$ in a way that operator $U$ satisfies:
\begin{equation}
\label{eq:extension1}
    U_i\,\widetilde{F_G}(i)\,U_i^\dagger=\widetilde{F_{G'}} (i)
\end{equation}
for all $i\in N\cap M$. 
\end{lemma}

\begin{proof}
    We shall consider two separate cases, either $M$ is of Type I, or satisfies $M\cap N\neq \emptyset$. 
    Firstly, assume $M$ is of Type I. Then, there are unique elements $s\in \St_{\ket{G}}^M$ and $s'\in \St_{\ket{G'}}^M$ in the reduced stabilizer groups. In accordance to \cref{prop:necessary_conditions_MLS}, we have $UsU^\dagger=s'$, in particular
    \begin{equation}
        U_i\,s_{\downarrow i}\,U_i^\dagger=s'_{\downarrow i}
    \end{equation}
    where $s_{\downarrow i}$ and $s_{\downarrow i}'$ denotes $i$-th element in a Pauli string $s$ and $s'$ respectively. We define functions $\widetilde{F_G},\widetilde{F_{G'}} :N\cap M\rightarrow\{X,Y,Z\}$ in the following way:
    \begin{align}
    \label{eq:cases}
        \widetilde{F_G}(i):=&
        \begin{cases}
            F_{G}(i) &\text{for } i\in N \\
            s_{\downarrow i} &\text{for } i\in M/N 
        \end{cases}
        \\
        \nonumber
        \widetilde{F_{G'}}(i):=&
        \begin{cases}
            F_{G'}(i) &\text{for } i\in N \\
            s_{\downarrow i}' &\text{for } i\in M/N 
        \end{cases}
        .
    \end{align}
    Notice that they extend functions $F_G,F_{G'}$, and satisfy (\ref{eq:extension1}) for all $i\in N\cap M$.  

    Secondly, suppose $M$ is of Type II and we have $M\cap N\neq \emptyset$. Fix arbitrary element $j\in N\cap M$, and consider Pauli matrices $F_G(j)$ and $F_{G'}(j)$. Notice that there are unique elements $s\in \St_{\ket{G}}^M$ and $s'\in \St_{\ket{G'}}^M$ in the reduced stabilizer groups such that $s_{\downarrow j}=F_G(j)$ and $s_{\downarrow j}'=F_{G'}(j)$. By \cref{prop:necessary_conditions_MLS}, operator $U$ is mapping elements form $ \St_{\ket{G}}^M$ to $\St_{\ket{G'}}^M$ $UsU^\dagger=s'$. Moreover, as $U_j\,s_{\downarrow j}\,U_j^\dagger=s'_{\downarrow j}$, we must have $U\,s\,U^\dagger=s'$. Therefore, we can extend functions $F_G,F_{G'}$ to $\widetilde{F_G},\widetilde{F_{G'}} :N\cap M\rightarrow\{X,Y,Z\}$ also by \cref{eq:cases}. Notice that they extend functions $F_G,F_{G'}$, and satisfy (\ref{eq:extension1}) for all $i\in N\cap M$. 
\end{proof}

The following Result shows how to extend compatible functions when the extension might not be unique. 

\begin{lemma}
\label{lemma:extension2}
Consider two graphs $G,G'$ sharing the same MLS cover $\mathcal{M}$. Moreover, let $N\subset [n]$ and suppose there exists LU operator $U=U_1\otimes\cdots\otimes U_n$, transforming $\ket{G}$ to $\ket{G'}$ which satisfies
\begin{equation}
    U_i\,F_G(i)\,U_i^\dagger=F_{G'} (i)
\end{equation}
for some functions $F_G,F_{G'}:N\rightarrow\{X,Y,Z\}$, and all $i\in N$. 

Consider $M\in \mathcal{M}$ which is of Type II and satisfies $M\cap N= \emptyset$. Then, we can extend functions $F_G,F_{G'}$ in three different ways: $\widetilde{F_G^k},\widetilde{F_{G'}^k} :N\cap M\rightarrow\{X,Y,Z\}$ where $k=1,2,3$ in a way that operator $U$ satisfies:
\begin{equation}
\label{eq:extension_k}
    U_i\,\widetilde{F_G^k}(i)\,U_i^\dagger=\widetilde{F_{G'}^k} (i)
\end{equation}
for some value $k=1,2,3$ and all $i\in N\cap M$. 
\end{lemma}

\begin{proof}
As we discussed in \cref{subsec:localset}, the reduced stabilizer groups $ \St_{\ket{G}}^M$ and $\St_{\ket{G'}}^M$ has three non-trivial elements each. Denote them by $s_1,s_2,s_3\in \St_{\ket{G}}^M$ and $s_1',s_2',s_3'\in \St_{\ket{G'}}^M$ respectively. Suppose operator $U$ transforms $\ket{G}$ to $\ket{G'}$. By \cref{prop:necessary_conditions_MLS}, there is a bijection $f_M: \St_{\ket{G}}^M \rightarrow \St_{\ket{G'}}^M$ such that $f_M (s) = U s U^\dagger $. Notice that there are exactly three possibilities, either
\begin{equation}
    f_M(s_1)=s_1', \text{  or}\,
        f_M(s_1)=s_2', \text{  or}\,
            f_M(s_1)=s_3'
            \label{eq:orr}
\end{equation}
Therefore, we can extend functions $F_G,F_{G'}$ to $\widetilde{F_G},\widetilde{F_{G'}} :N\cap M\rightarrow\{X,Y,Z\}$ in three different ways:
\begin{align}
    \label{eq:cases2}
        \widetilde{F_G^k}(i):=&
        \begin{cases}
            F_{G}(i) &\text{for } i\in N \\
            s_{1\downarrow i} &\text{for } i\in M/N 
        \end{cases}
        \\
        \nonumber
        \widetilde{F_{G'}^k}(i):=&
        \begin{cases}
            F_{G'}(i) &\text{for } i\in N \\
            s_{k\downarrow i}' &\text{for } i\in M/N 
        \end{cases}
    \end{align}
for $k=1,2,3$. Notice that $   \widetilde{F_G^k}$ does not depends on $k$, while $\widetilde{F_{G'}^k}$, indeed, depends on $k$. By (\ref{eq:orr}), we know that for at least one $k=1,2,3$, we have (\ref{eq:extension_k}) for all $i\in N\cap M$, which finishes the proof. 
\end{proof}

Equipped with \cref{lemma:extension} and \cref{lemma:extension2}, we can now proceed with the proof of \cref{prop_pairs_of_func}.

\begin{proof}[Proof of \cref{prop_pairs_of_func}]
    We shall prove the correctness of Algorithm 1 presented on \cref{fig:algo_subroutines}. Consider two graphs $G$ and $G'$ that share the same MLS cover $\mathcal{M}$. Notice that in the first step Algorithm 1 checks the necessary conditions for $\ket{G}$ and $\ket{G'}$ to be LU-equivalent. Those conditions are based on \cref{lemma:overlaps_conditions}. 

    In subsequent steps, Algorithm 1 defines the family of pairs of functions $F_G^\ell, F_{G'}^\ell$ indexed by parameter $\ell$ such that any LU operator $U=U_1\otimes\cdots\otimes U_n$ satisfies $U_i\,F_G^\ell (i) \,U_i^\dagger = F_{G'}^\ell (i) $for at least one parameter $\ell$ and all $i\in [n]$. It does it recursively, by extending domain $N$ on which functions $ F_G^\ell, F_{G'}^\ell :N\rightarrow\{X,Y,Z\}$ are defined. It starts with a trivial domain $N=\emptyset$. In each recursive step it takes $M\in \mathcal{M}$ and extends it to bigger domain $N\cap M$. As $\mathcal{M}$ is MLS cover, at the end functions $ F_G^\ell, F_{G'}^\ell $ are defined on $N=[n]$. In each extension step, the algorithm searches for $M\in \mathcal{M}$ either of Type I or satisfying $M\cap N\neq \emptyset$. In such cases, the extension procedure is unique, which based on \cref{lemma:extension}. When $M\in \mathcal{M}$ with either of those properties cannot be find, the algorithm extends functions in three different ways, which is based on \cref{lemma:extension2}. 
    
    As we used \cref{lemma:extension} and \cref{lemma:extension2} while extending functions $ F_G^\ell, F_{G'}^\ell $ to the whole domain $[n]$, we can conclude that for arbitrary LU operator $U=U_1\otimes\cdots\otimes U_n$ satisfies $U_i\,F_G^\ell (i) \,U_i^\dagger = F_{G'}^\ell (i) $ for at least one parameter $\ell$ and all $i\in [n]$. Moreover, while extending functions $ F_G^\ell, F_{G'}^\ell $, the algorithm  extended functions in three different ways (based on \cref{lemma:extension2}) exactly $t$ times, where $t$ denotes the number of connected components in the intersection graph of $\mathcal{M}$ consisting entirely of MLSs of Type II. 
\end{proof}

\section{Observations regarding transitional matrices} 
\label{appendix:0}

In this section we provide an exact formula and make some technical observations regarding transitional matrices. 
Firstly, for arbitrary $A \in \{ X, Y, Z\} $, we define:
\begin{align}
    \mathbf{C} (A, A ):=\Id ,
    \label{eq:transition_first}
\end{align}
furthermore, we define:
\begin{align}
    \mathbf{C} (X,Y )=\mathbf{C} (Y,X ) &:= \tfrac{1}{\sqrt{2}}\begin{psmallmatrix}1 & -1\\i & i\end{psmallmatrix},
\nonumber 
\\
 \mathbf{C} (X,Z )=\mathbf{C} (Z,X ) &:= \tfrac{1}{\sqrt{2}}\begin{psmallmatrix}1 & 1\\1 & -1\end{psmallmatrix}=H,
    \label{eq:transition_second}
    \\
    \mathbf{C} (Y,Z)=\mathbf{C} (Z,Y ) &:= \tfrac{1}{\sqrt{2}}\begin{psmallmatrix}1 & i\\i & 1\end{psmallmatrix}.
\nonumber 
\end{align}
Notice that \cref{eq:transition_first,eq:transition_second} together define transition matrices $\mathbf{C} (A,B)$ for all possible pairs of matrices $A,B \in \{ X, Y, Z\} $ which satisfies required properties. Indeed, a direct calculation shows the following. 

\begin{lemma}
\label{lemma:XXX}
For arbitrary Pauli matrices $A,B,D\in \{ X, Y, Z\} $, we have the following:
\begin{equation}
    \mathbf{C} (A,B)^\dagger\, D \,\mathbf{C} (A,B) =
    \begin{cases}
        A &\text{if }D=B,\\
        B &\text{if }D=A, \\
        +D &\text{if }A=B, \\
        - D &\text{otherwise. }
    \end{cases}
    .
\end{equation}
\end{lemma}


We proceed to the technical lemma that classifies all unitary matrices that transform a given Pauli matrix $A\in \{ X, Y, Z\} $ into another given Pauli matrix $B \in \{ X, Y, Z\} $.

\begin{lemma}
\label{lemma:form_preserving_pauli}
Suppose that the unitary matrix $V$ preserves Pauli-$Z$ matrix up to the phase, i.e. $V Z V^\dagger=\omega Z$. 
Then, the matrix $V$ is of the following form:
\begin{equation}
V=
 X^s Z^\alpha
\end{equation}
for some integer $s\in\{0,1\}$, and real $\alpha\in [0,2]$, and $|\omega |=1$ is a phase. 
\end{lemma}

\begin{proof}
Every unitary matrix $V$ that preserves Pauli-$Z$ matrix up to the phase, i.e. $Z=\omega V Z V^\dagger$ is of the form $V=\omega X^s Z^\alpha$ for some $s\in\{0,1\}$, parameter $\alpha \in [0,2]$ and phase $\omega$. i.e. $|\omega |=1$. Indeed, this can be seen through direct calculation or by interpreting the operator $V$ as a rotation of the Bloch sphere. The only rotation that preserves the $Z$-axis is a rotation around $Z$ by an arbitrary angle, which may or may not be followed by a flip of the poles. Notice that this can also be achieved by $V = Y Z^\alpha$, but this is redundant since $Y Z^\alpha = X Z^{2-\alpha}$. 
\end{proof}

\begin{lemma}
\label{ob:Auxi_2.5_pm}
Consider matrix 
\begin{equation}
\label{eq:ui2app}
\breve{U} =
Z^{t} {\mathbf{C} (Y,X)^{\ell}}^\dagger
(X^{s} Z^{\alpha})
Z^{t'} \mathbf{C} (Y,X)^{\ell'}
,
\end{equation}
for some $t,t',\ell,\ell', s\in\{0,1\}$
and continuous parameter $\alpha\in [0,2]$. 
Regardless of the values $t,t',\ell,\ell'$, (\ref{eq:ui2app}) is equal to 
\begin{equation}
\label{eq:ui3app}
\breve{U}_i =X^{s'} Z^{\alpha'},
\end{equation}
for some $s'\in\{0,1\}$
and continuous parameter $\alpha'\in [0,2]$ and some phase $|\omega_i|=1$. 

Similarly, matrix 
\begin{equation}
\label{eq:ui2app22}
\breve{U} =
X^{t} 
(X^{s} Z^{\alpha})
X^{t'} 
,
\end{equation}
for some $t,t' s\in\{0,1\}$ and continuous parameter $\alpha\in [0,2]$ can be written as
\begin{equation}
\label{eq:ui3appgg}
\breve{U}_i =X^{s'} Z^{\alpha'},
\end{equation}
for some $s'\in\{0,1\}$ and continuous parameter $\alpha'\in [0,2]$ and some phase $|\omega_i|=1$. 
\end{lemma}

\begin{proof}
    Direct calculation. 
\end{proof}

\section{Normal form of stabilizer generators} 
\label{appendix:A}

In this section we demonstrate that arbitrary stabilizer state is generated by generators with the check matrix in a normal form. Furthermore, we show that any stabilizer state\footnote{With some assumption on a phase vector, see \cref{lemma:Auxi_2_pm}} is locally LC equivalent to another stabilizer state with the check matrix in a strong normal form and trivial phase vector. Furthermore, we show what it is the exact form of such LU transformation.  
Those technical results are used in the proof of \cref{prop:reduction_from_given_MLS}.


\begin{lemma}
\label{lemma:Auxi_0_pm}
Up to permutation of qubits, any stabilizer state is generated by generators with the check matrix in a normal form.
\end{lemma}

\begin{proof}
Consider a stabilizer state $\ket{\Sg}$, and arbitrary set of generators $\Sg$. 
Furthermore, consider a check matrix $[\X_\Sg |\Z_\Sg ]$ associated to a stabilizer state $\ket{\Sg}$ generated by $\Sg$. 
Recall that simultaneous row-elimination operators in matrices $\X_\Sg$ and $\Z_\Sg$ correspond to change of the stabilizer generator $\Sg$. 
Such operations do not, however, change the stabilizer state $\ket{\Sg}$. 
Furthermore, the simultaneous permutation or rows in matrices $\X_\Sg$ and $\Z_\Sg$ correspond to permutation of qubits in $\ket{\Sg}$. 
In the following, we present the sequence of row-operations and column permutations that transforms given check matrix into its normal form. 
By this, we show, that up to permutation of qubits, any stabilizer state has a generators in a normal-form (\ref{eq:stabilizer_generator_normal_form}). 
As the procedure is a variant of Gauss elimination, it can be done in $\mathcal{O}(n^3)$ time. 

Suppose that $\rk \X_\Sg =k$. 
Gauss elimination and permutation of columns applied on rows of $\X_\Sg$ can transform a check matrix to the following form
\[
\X_\Sg=
\begin{bNiceArray}{c|c}
\Block{1-1}{\id_k}&\X_{k|n-k}^{}\\
\cline{1-2}
\Block{1-1}{0}&\Block{1-1}{0}\\
\end{bNiceArray},
\quad
\Z_\Sg=
\begin{bNiceArray}{c|c}
\Block{1-1}{\Z_{k|k}}&\Block{1-1}{\Z_{k|n-k}}\\
\cline{1-2}
\Block{1-1}{\Z_{n-k|k}}&\Block{1-1}{\Z_{n-k|n-k}}\\
\end{bNiceArray},
\]
Furthermore, suppose that $\rk \Z_{n-k|n-k} =\ell$. 
By applying Gauss elimination on the latter $n-k$ rows and permutation of latter $n-k$ columns, we can further transform the check matrix to the following form:
\begin{equation}
\label{eq:XZprim}
\X_\Sg'=
\begin{bNiceArray}{c|c}
\Block{1-1}{\id_k}&\X_{k|n-k}'\\
\cline{1-2}
\Block{1-1}{0}&\Block{1-1}{0}\\
\end{bNiceArray},
\quad
\Z_\Sg'=
\begin{bNiceArray}{c|c|c}
\Block{1-1}{\Z_{k|k}}&\Block{1-2}{\Z_{k|n-k}'}\\
\cline{1-3}
\Z_{\ell|k}' &\Block{1-1}{\id_\ell}&\Z_{\ell |n-k-\ell}'\\
\cline{1-3}
\Block{1-1}{\Z_{n-k-\ell |k}'}&\Block{1-2}{0}\\
\end{bNiceArray}.
\end{equation}

Suppose now, that $\ell \neq n-k$. 
Furthermore, consider, the stabilizer generator $s_n'$ corresponding to the last row in $[\X_\Sg'|\Z_\Sg']$. 
Notice that such generator $s_n'$ is a Pauli string that consists only of identity matrices and Pauli-$Z$ matrices on some of the first $k$ positions. 
As $s_n'$ is generator, it contains Pauli-$Z$ matrix on some position $k_1$, for $1\leq k_1\leq k$. On the other hand, consider generator $s_{k_1}'$ corresponding to the $k_1$ row in $[\X_\Sg'|\Z_\Sg']$.  
Notice that $s_{k_1}'$ is a Pauli string with either Pauli-$X$ or Pauli-$Y$ matrix on position $k_1$ and with only identity and Pauli-$Z$ matrices on any other poition form $1$ to $k$. 
Notice that $s_n'$ contains Pauli-$Z$ matrix on position $k_1$ while $s_{k_1}'$ contains Pauli-$X$ or Pauli-$Y$ matrix on that position. On any other position $s_n'$ and $s_{k_1}'$ contain commuting Pauli matrices. 
Therefore $s_n'$ and $s_{k_1}'$ are not commuting Pauli strings, hence cannot form a stabilizer generators. 
In that way, we have shown that $\ell = n-k$, matrix $\Z_{n-k|n-k}$ is of the full rank, and hence (\ref{eq:XZprim}) is of the following form: 
\begin{equation}
\label{eq:XZprim2}
\X_\Sg'=
\begin{bNiceArray}{c|c}
\Block{1-1}{\id_k}&\X_{k|n-k}'\\
\cline{1-2}
\Block{1-1}{0}&\Block{1-1}{0}\\
\end{bNiceArray},
\quad
\Z_\Sg'=
\begin{bNiceArray}{c|c}
\Block{1-1}{\Z_{k|k}}&\Block{1-1}{\Z_{k|n-k}'}\\
\cline{1-2}
\Z_{n-k|k}' &\Block{1-1}{\id_{n-k}}\\
\end{bNiceArray}.
\end{equation}

By multiplying first $k$ rows by the latter $n-k$ rows, we can transform matrix ${\Z_{k|n-k}'}$ into zero matrix, hence the check matrix transforms to the following form:
\begin{equation}
\label{eq:XZbis}
\X_\Sg''=
\begin{bNiceArray}{c|c}
\Block{1-1}{\id_k}&\X_{k|n-k}''\\
\cline{1-2}
\Block{1-1}{0}&\Block{1-1}{0}\\
\end{bNiceArray},
\quad
\Z_\Sg''=
\begin{bNiceArray}{c|c}
\Block{1-1}{\Z_{k|k}}&\Block{1-1}{0}\\
\cline{1-2}
\Z_{n-k|k}' &\Block{1-1}{\id_{n-k}}\\
\end{bNiceArray}.
\end{equation}
It remains to check the condition imposed on matrices $\X_{k|n-k}''$, $\Z_{k|k}$, and $\Z_{n-k|k}'$ by the commutation relations between generators. 
Firstly, notice that for arbitrary $i,j\leq k$, the commutation of generators $s_i'' $ and $s_j''$ impose that $z_{ij}=z_{ji}$ in matrix $\Z_{k|k}$. 
That shows that matrix $\Z_{k|k}$ is symmetric. 
Secondly, for arbitrary $i\leq k$, and $j>k$, the commutation of generators $s_i'' $ and $s_j''$ impose that $x_{ij}=z_{ji}$ where $x_{ij}$ is $ij$-element of a matrix $\X_{k|n-k}''$ and $Z_{ji}$ is $ji$-element of a matrix $\Z_{n-k|k}''$. 
This shows that $\Z_{n-k|k}''=(\X_{k|n-k}'')^\text{T}$. 
As a result, (\ref{eq:XZbis}) is, indeed, a normal form.
\end{proof}


It is easy to see that stabilizer generators written in a normal form with phase vectors $\pm 1$ might be transformed to stabilizer generators written in a strong normal form with phase vectors $ 1$. Indeed, we have the following observation, which follows from a direct calculation.

\begin{lemma}
\label{lemma:Auxi_2_pm}
Consider stabilizer generators $\Sg=\langle s_1,\ldots,s_n\rangle$ with the corresponding matrix in a normal form, i.e.
\begin{equation}
\label{eq:sghat44}
\X_{\Sg}=
\begin{bNiceArray}{c|c}
\Block{1-1}{\id_{k}}&\Gamma_{k|n-k_1}^{}\\
\cline{1-2}
\Block{1-1}{0}&\Block{1-1}{0}\\
\end{bNiceArray},
\quad
\Z_{\Sg}=
\begin{bNiceArray}{c|c}
\Block{1-1}{\Gamma_{k|k}}&\Block{1-1}{0}\\
\cline{1-2}
\Block{1-1}{\Gamma_{k|n-k}^{\text{T}}}&\Block{1-1}{\id_{n-k}}\\
\end{bNiceArray}, 
\end{equation}
for some symmetric matrix $\Gamma_{k|k}^{}$  and arbitrary matrix $\Gamma_{k|n-k}^{}$, and assume that phases of stabilizer generators $s_i$ are $\phi_{s_i}=\pm 1$. 
Consider the following LC-operator
\[
W=
\prod_{i=1}^k
Z^{\phi_{s_i}} \mathbf{C} (Y,X)^{(\Gamma_{k|k})_{ii}}
\prod_{i=k+1}^n
X^{\phi_{s_i}} 
\] 
where $(\Gamma_{k|k})_{ii}$ is the $ii$-entry of a matrix $\Gamma_{k|k}$ and $\phi_{s_i}$ is a phase of stabilizer $s_i$. 
Operator $W$ transforms generators in $\Sg$ into generators $\Sg'=\langle s_1',\ldots,s_n'\rangle$, $s_i':=W^\dagger s_i W$ such that the check matrix corresponding to $\Sg'$ is in a strong normal form and phase vectors $\phi_{s_i'}=1$. 
The check-matrix for generators $\Sg'$ is of the form (\ref{eq:sghat44}) with removed all non-zero entries in a matrix $\Gamma_{k|k}$. 
\end{lemma}

\begin{proof}
    Direct calculation. 
\end{proof}



\begin{lemma}
\label{lemma:Auxi_3.5_pm}
Consider a generators $\Sg=\langle s_1,\ldots,s_n\rangle$ of a stabilizer state $\ket{\Sg}$ with the corresponding check matrix $[\X_{\hat{\Sg}} | \Z_{\hat{\Sg}} ]$ in the normal form, i.e.
\begin{equation}
\label{eq:sghat}
\X_{\Sg}=
\begin{bNiceArray}{c|c}
\Block{1-1}{\id_{k}}&\Gamma_{k|n-k}^{}\\
\cline{1-2}
\Block{1-1}{0}&\Block{1-1}{0}\\
\end{bNiceArray},
\quad
\Z_{\Sg}=
\begin{bNiceArray}{c|c}
\Block{1-1}{\Gamma_{k|k}}&\Block{1-1}{0}\\
\cline{1-2}
\Block{1-1}{\Gamma_{k|n-k}^{\text{T}}}&\Block{1-1}{\id_{n-k}}\\
\end{bNiceArray}, 
\end{equation}
for some symmetric and off-diagonal matrix $\Gamma_{k|k}^{}$  and arbitrary matrix $\Gamma_{k|n-k}^{}$. 
Suppose that the stabilizer state $\ket{\Sg}$ is LU equivalent to another stabiliser state $\ket{\Sg'} $, i.e.
\[
\widetilde{U}\ket{\Sg}=\ket{\Sg'}
\]
by some phase operator $\widetilde{U}$, i.e. $\widetilde{U}\ket{I}=\omega_I \ket{I}$ for all computational basis vectors $\ket{I}$, i.e. $I\in \{0,1\}^n$ and some phases $|\omega_I|=1$. 
Then there exists generators $\Sg'=\langle s_1',\ldots,s_n'\rangle$ of a stabilizer $\Sg' $, such that the corresponding check matrix $[\X_{\hat{\Sg'}} | \Z_{\hat{\Sg'}} ]$ is in the normal form with the same  $\X$-part as for generators $\Sg$, i.e. 
\begin{equation}
\label{eg:st6}
\X_{\Sg'}=
\begin{bNiceArray}{c|c}
\Block{1-1}{\id_{k}}&\Gamma_{k|n-k}\\
\cline{1-2}
\Block{1-1}{0}&\Block{1-1}{0}\\
\end{bNiceArray},
\quad
\Z_{\Sg'}=
\begin{bNiceArray}{c|c}
\Block{1-1}{\Gamma_{k|k}'}&\Block{1-1}{0}\\
\cline{1-2}
\Block{1-1}{{\Gamma_{k|n-k}}^{\text{T}}}&\Block{1-1}{\id_{n-k}}\\
\end{bNiceArray},
\end{equation}
for some symmetric matrix $\Gamma_{k|k}^{}$ and the same matrix $\Gamma_{k|n-k}$ as in (\ref{eq:sghat}).  
\end{lemma}

\begin{proof}
Consider stabilizer $\Sg$ and a corresponding check-matrix presented on \cref{eq:sghat}. 
Notice that for any stabilizer state $\X$-part of the check-matrix uniquely defines its support, in particular the support of $\ket{\Sg}$ is given by
\begin{align}
\label{eq:st}
\supp \big( \ket{\Sg} \big) &=
\{\ket{I\cdot \X_\Sg } :{I\in\{0,1\}^n}\} 
\\&=
\{\ket{J\cdot [\Id_k |\Gamma_{k|n-k} ] } :{J\in\{0,1\}^k}\},
\nonumber
\end{align}
hence 
\begin{equation}
\label{eq:www}
\Big\vert \supp \big( \ket{\Sg} \big) \Big\vert =
2^{\rk \X_\Sg } = 2^k.
\end{equation}
Consider now arbitrary set of stabilizer generators $\Sg'=\langle s_1',\ldots,s_n'\rangle$ of a stabilizer $\Sg' $ and corresponding check matrix $[\X_{\Sg'}|\Z_{\Sg'}]$. 
We have 
\begin{align*}
\supp \big( \ket{\Sg'} \big) &=
\{\ket{I\cdot \X_{\Sg'} } :{I\in\{0,1\}^n}\} ,
\end{align*}
and hence $|\supp ( \ket{\Sg'})| =2^{\rk \X_{\Sg'} }$. 
As phase operator does not change the support of a state, we have
\begin{equation}
\label{eq:st2}
\supp \big( \ket{\Sg} \big) =\supp \big( \ket{\Sg'} \big),
\end{equation}
and by combining it with (\ref{eq:www}), we deduce that $\rk \X_{\Sg'}=\rk \X_{\Sg}=k$. 
Without loss of generality assume that the first $k$ rows of $\X_{\Sg'}$ are linearly independent (this can be achieved by permuting stabilizer generators $s_i'$). 
Notice that
\begin{align*}
\supp \big( \ket{\Sg'} \big) &=
\{\ket{J\cdot \X_{\Sg'}^{k|n} } :{J\in\{0,1\}^k}\} ,
\end{align*}
where $\X_{\Sg'}^{k|n}$ is part of $\X_{\Sg'}$ containing first $k$ rows. 
In accordance to (\ref{eq:st}) and (\ref{eq:st2}), we have 
\begin{align}
\{\ket{J\cdot [\Id_k |\Gamma_{k|n-k} ] } :{J\in\{0,1\}^k}\}=
\\
\nonumber
=
\{\ket{J\cdot \X_{\Sg'}^{k|n} } :{J\in\{0,1\}^k}\} 
\end{align}
and hence there exists an invertible $k\times k$ operator $\mathcal{O}$ such that
\begin{equation}
\label{eq:st3}
[\Id_k |\Gamma_{k|n-k} ]=\mathcal{O} \cdot \X_{\Sg'}^{k|n} .
\end{equation}
Consider set of stabilizers $s_i''$ defined as 
\begin{equation}
s_i'':=
\begin{cases}
  \prod_{j=1}^k o_{ij}s_j'  & \text{ for }i\leq k  \\
  s_i'  & \text{ for }i> k
\end{cases}
\end{equation}
where $o_{ij}$ are entries of matrix $\mathcal{O}^{-1}$. 
Notice that, as the matrix $\mathcal{O}$ is invertible, the set $\Sg''=\langle s_1'',\ldots,s_n''\rangle$  generates stabilizer $\Sg'$. 
Furthermore, in accordance to (\ref{eq:st3}) the $\X$ part of the check matrix corresponding to the set $\Sg''$ is of the form
\begin{equation}
\label{eq:st4}
\X_{\Sg''}=
\begin{bNiceArray}{c|c}
\Block{1-1}{\id_{k}}&\Gamma_{k|n-k}\\
\cline{1-2}
\Block{1-2}{A}\\
\end{bNiceArray},
\end{equation}
for some $n-k\times n$ matrix $A$. 
As $\Sg''$ and $\Sg'$ generate the same stabilizer, the rank $\rk \X_{\Sg''}=k$ and rows in matrix $A$ are linear combinations of rows in $[{\id_{k}}|\Gamma_{k|n-k}]$. 
Consider, the following set of stabilizers $s_i'''$ defined as 
\begin{equation}
s_i''':=
\begin{cases}
  s_i'' & \text{ for }i\leq k  \\
  s_i''\cdot \prod_{j=1}^k a_{ij}s_j''   & \text{ for }i> k
\end{cases}
\end{equation}
where $a_{ij}$ are entries of matrix $A$. 
Notice that the set $\Sg'''=\langle s_1''',\ldots,s_n'''\rangle$ generates stabilizer $\Sg''$ (and hence $\Sg'$). Furthermore, the $\X$ part of the check matrix corresponding to the set $\Sg'''$ is of the form
\begin{equation}
\label{eq:st4P}
\X_{\Sg'''}=
\begin{bNiceArray}{c|c}
\Block{1-1}{\id_{k}}&\Gamma_{k|n-k}\\
\cline{1-2}
\Block{1-2}{0}\\
\end{bNiceArray},
\end{equation}
and hence it is the same as $\X$-part of the check matrix corresponding to stabilizer generators $\Sg$, see (\ref{eq:sghat}). 

Following exactly the same steps as presented in the proof of \cref{lemma:Auxi_0_pm}, one can find another stabilizer generators of stabilizer $\Sg'''$, with the check matrix being of the form (\ref{eg:st6}), which finishes the proof.
\end{proof}

\begin{lemma}
\label{lemma:Auxi_4_pm}
Consider stabilizer generators $\hat{\Sg}=\langle \hat{s}_1,\ldots,\hat{s}_n\rangle$ with the corresponding check matrix $[\X_\Sg | \Z_\Sg ]$ in the strong normal form:
\begin{equation}
\label{auxii}
\X_\Sg=
\begin{bNiceArray}{c|c}
\Block{1-1}{\id_k}&\Gamma_{k|n-k}^{}\\
\cline{1-2}
\Block{1-1}{0}&\Block{1-1}{0}\\
\end{bNiceArray},
\quad
\Z_\Sg=
\begin{bNiceArray}{c|c}
\Block{1-1}{\Gamma_{k|k}}&\Block{1-1}{0}\\
\cline{1-2}
\Block{1-1}{\Gamma_{k|n-k}^{\text{T}}}&\Block{1-1}{\id_{n-k}}\\
\end{bNiceArray},
\end{equation}
for some symmetric and off-diagonal matrix $\Gamma_{k|k}^{}$ and arbitrary $k\times n-k$ matrix $\Gamma_{k|n-k}^{}$, and with trivial corresponding phase vectors $\phi_{\hat{s}_i}=1$. 
Then by applying LU-operator:
\[
\hat{V}= \id^{\otimes k} \otimes H^{\otimes n-k},
\]
the stabilizer state $\ket{\hat{\Sg}}$ is transformed into a graph state $\ket{\hat{G}}:=\hat{V}\ket{\hat{\Sg}}$ corresponding to an adjacency matrix 
\begin{equation}
\label{auxiii}
\Gamma_{\hat{G}}=
\begin{bNiceArray}{c|c}
\Block{1-1}{\Gamma_{k|k}}&\Gamma_{k|n-k}^{}\\
\cline{1-2}
\Block{1-1}{\Gamma_{k|n-k}^{\text{T}}}&\Block{1-1}{0}\\
\end{bNiceArray}.
\end{equation}
In particular, graph $\hat{G}$ does not have any edges between the last $n-k$ vertices. 
\end{lemma}

\begin{proof}
Notice that generators $\hat{s}_i$ on the last $n-k$ places have either $Z$ or $X$ Pauli matrices only. 
As Hadamard matrix exchange those matrices, i.e. $H= \mathbf{C} (X,Z)=\mathbf{C} (Z,X)$, the operator $\hat{V}$ exchanges the last $n-k$ columns (\ref{auxii}) of $\X_\Sg$ with those in $\Z_\Sg$.  Notice that the phase vector remains unchanged. 
As a result, operator $\hat{V}$ transforms $\hat{\Sg}$ into a stabilizer state with the check-matrix: $[\id^n | \Gamma_{\hat{G}} ]$, where $\Gamma_{\hat{G}}$ is presented on (\ref{auxiii}). Notice that this is, indeed a graph state, with corresponding adjacency matrix presented on (\ref{auxiii}).
\end{proof}

\section{Proof of \cref{corollary:all_are_that_form} and \cref{Th:reduction}}
\label{appendix:B}

In this section, we prove \cref{corollary:all_are_that_form} and \cref{Th:reduction}. 

We begin with technical lemma. Recall that in \cref{sec:StepII}, we introduced weighted hypergraph states. In particular, we discussed that the action of $\prod_{i=k+1}^n X_i^{\alpha_{i}}$ on disconnected vertices for graph states, see \cref{eq:X_mluti_action}, and $\prod_{i=1}^k
Z_i^{\alpha_i}$ on weighted hypergraph states, see \cref{eq:Z_action}. We shall use those two facts to prove the following technical result.

\begin{lemma}
\label{prop:solution_1}
Consider two graphs $\hat{G}$ and $\hat{G'}$ correspond to the adjacency matrices (\ref{auxiii_main}), and assume that the corresponding graph states $|\hat{G}\rangle$ and $|\hat{G'}\rangle$ are LU equivalent by the operator 
\begin{equation}
\label{eqww}
\hat{U}'= \prod_{i=1}^k
X_i^{s_i'}Z_i^{\alpha_i'}\otimes
\prod_{i=k=1}^n
Z_i^{s_{i}'}X_i^{\alpha_{i}'},
\end{equation}  
for some values $\alpha_i' \in [0,2]$ and $s_i'\in\{0,1\}$. 
Then states $|\hat{G}\rangle$ and $|\hat{G'}\rangle$ are LU equivalent by the operator 
\begin{equation}
\label{eq:ZX_2}
\hat{U}=\prod_{i=1}^k
Z_i^{\alpha_i}\otimes
\prod_{i=k+1}^n
X_i^{\alpha_{i}},
\end{equation}
for some values $\alpha_i \in [0,2]$. 
\end{lemma}

\begin{proof}
Notice that in graph $\hat{G}$ there are no connections between vertices $k+1,\ldots,n$. 
Therefore, using \cref{eq:X_mluti_action}, we can deduce that 
\begin{align}
\ket{\hat{H_w}}:=\id^k\otimes \prod_{i=k+1}^n
X_i^{\alpha_{i}'} \ket{\hat{G}}=
\ket{\prod_{i=k+1}^n C_{\delta_i,\alpha_i}\,\hat{G}}
\label{eq:kkk}
\end{align}
where $\delta_i$ is a neighborhood of $i$, and $\hat{H_w}$ is a weighted hypergraph obtained from hypergraph $\hat{H}$ by adding the weights $(-2)^{|e|-1}\alpha$ to all hyperedges $e\subset \delta_i$ for all $i=k+1,\ldots,n$.

Furthermore, form \cref{eq:Z_action}, we have that 
\begin{align}
\ket{\hat{H_w}'}:= \prod_{i=1}^k Z_i^{\alpha_1'}\otimes \prod_{i=k+1}^n
Z_i^{s_i'} \ket{\hat{H_w}}
\label{eq:kkk2}
\end{align}
where $\hat{H_w}'$ is a weighted hypergraph obtained from weighted hypergraph $\hat{H_w}$ by adding the weights $\alpha_i$ and $s_i$ to respective vertices.

Notice, that for $s\in\{0,1\}$, we have $(X^s)^{-1}=X^s$. Therefore, from (\ref{eqww}) the state $\ket{\hat{H_w}'}$ can be obtained from a graph state $\ket{\hat{G}'}$ as follows:
\begin{align}
\nonumber
\ket{\hat{H_w}'}&= \prod_{i=1}^k X^{s_i'} \otimes \id^{n-k}
\ket{\hat{G}'}
\end{align}
In particular, using \cref{eq:X_action} recursively to all $i\in [k]: s_i'=1$ we transform hypergraph states as follows:
\begin{align}
\nonumber
\ket{\hat{H_w}'}&= \prod_{i=1}^k X^{s_i'} \otimes \id^{n-k}
\ket{\hat{G}'} =
\prod_{i=1, s_i'=1}^k \prod_{j\in \delta_i} C_{\{j\},1} \ket{\hat{G}'} =
\\&
\prod_{i\in \Delta} Z_i\ket{\hat{G}'}
\label{eq:kkk3}
\end{align}
where $\Delta:=\triangle_{i=1, s_i'=1} \delta_i$ is a symmetric difference of the sets in $\{
\delta_i: s_i=1\}$, and hence $\hat{H_w}'$ is in fact a hypergraph state obtained from graph $\hat{G}'$ by adding $+1$ weights to all the single-vertex hyperedges in $\Delta$. Indeed, in (\ref{eq:kkk3}), we can sequentially apply operators $X^{s_i'}$ for $i=1,\ldots,k$, which in each step results in a hypergraph state, as all added weights are equal to $+1$.

As the graph $\hat{G}'$ is obtained from hypergraph $\hat{H_w}'$ by adding $+1$ weights to all the single-vertex hyperedges in $\Delta$, we have
\begin{align}
\ket{\hat{G}'}= 
\prod_{i\in \Delta} Z_i\ket{\hat{H_w}'}.
\label{eq:kkk4}
\end{align}
Combining (\ref{eq:kkk4}) with (\ref{eq:kkk2}), we have that 
\begin{align}
\ket{\hat{G}'}:= \prod_{i=1}^k Z_i^{\alpha_1'+s_i}\otimes \prod_{i=k+1}^n
Z_i^{s_i'+s_i} \ket{\hat{H_w}}
\label{eq:kkk5}
\end{align}
where we define $s_i=1$ iff $i\in \Delta$ and $s_i=0$ otherwise. 
Notice that as vertices $k+1,\ldots,n$ are pairwise not connected, operators $C_{\delta_i,\alpha_i}$ in (\ref{eq:kkk}) does not produces any single-vertex hyperedges in $\hat{H_w}$ on vertices $k+1,\ldots,n$.
Therefore $s_i'+s_i=0$ for all $i=k+1\ldots,n$. 
By combining (\ref{eq:kkk5}) with (\ref{eq:kkk5}), we have
\begin{align}
\ket{\hat{G}'}:= \prod_{i=1}^k Z_i^{\alpha_1}\otimes \prod_{i=k+1}^n
X_i^{\alpha_i'} \ket{\hat{H_w}}
\label{eq:kkk6}
\end{align}
where $\alpha_i:=\alpha_i+s_i$ which proves the statement. 
\end{proof}

Consequently, we establish the following result.

\begin{proposition}
\label{prop:reduction_from_given_MLS}
Consider graph states $\ket{G}$ and $\ket{G'}$ and functions $F_G,F_{G'}:[n]\rightarrow \{X,Y,Z\}$. If $\ket{G}$ and $\ket{G'}$ are LU-equivalent by the operator operator $U=U_1\otimes\cdots\otimes U_n$ such that 
\begin{align}
\label{eq:aux1}
F_{G'} (i) = U _i \,F_G(i) \,U_i^\dagger 
\end{align}
for arbitrary $i\in [n]$. Then, we can construct graphs $\ket{\hat{G}}$ and $\ket{\hat{G'}}$, such that 
\begin{equation}
\ket{\hat{G}} \stackrel{LC}{\cong} \ket{G},
\quad
\ket{\hat{G'}} \stackrel{LC}{\cong} \ket{G'},
\label{eq:Ghat_construction}
\end{equation}
and, up to permutation of qubits, $\hat{G}$ and $\hat{G'}$ correspond to the following adjacency matrices
\begin{equation}
\label{auxiii_main}
\Gamma_{\hat{G}}=
\begin{bNiceArray}{c|c}
\Block{1-1}{\Gamma_{k|k}}&\Gamma_{k|n-k}^{}\\
\cline{1-2}
\Block{1-1}{\Gamma_{k|n-k}^{\text{T}}}&\Block{1-1}{0}\\
\end{bNiceArray},\quad
\Gamma_{\hat{G'}}=
\begin{bNiceArray}{c|c}
\Block{1-1}{\Gamma_{k|k}'}&\Gamma_{k|n-k}^{}\\
\cline{1-2}
\Block{1-1}{\Gamma_{k|n-k}^{\text{T}}}&\Block{1-1}{0}\\
\end{bNiceArray}.
\end{equation} 
for some symmetric and off-diagonal matrices $\Gamma_{k|k}^{}$ and $\Gamma_{k|k}'$ and arbitrary $k\times n-k$ matrix $\Gamma_{k|n-k}^{}$,  
and states $\ket{\hat{G}}$ and $\ket{\hat{G'}}$ are LU equivalent by the operator 
\begin{equation}
\label{eq:ZX_new}
\hat{U}=
\prod_{i=1}^k 
Z^{\alpha_i}_i
\prod_{i=k+1}^n 
X^{\alpha_i}_i
\end{equation}
for some values $\alpha_i \in [0,2]$. 
Furthermore, the construction of graphs $\hat{G}$ and $\hat{G'}$ 
is algorithmic and presented in the proof of the statement, and runs in $\mathcal{O}(n^3)$ time in number of vertices in $G$ and $G'$.
\end{proposition}

\begin{proof}
\label{remark:Ghat}
Consider stabilizer state $\ket{S}:=V\ket{G}$ obtained from a graph state $\ket{G}$ by applying the LC-operation $V=V_1\otimes\cdots\otimes V_n$ where
\begin{equation}
\label{eq:a1}
V_i=\mathbf{C} (F_G(i),Z),
\end{equation}
and another stabilizer state $\ket{S'}:=V'\ket{G'}$ obtained from a graph state $\ket{G'}$ by applying the LC-operation $V'=V_1'\otimes\cdots\otimes V_n'$ where
\begin{equation}
\label{eq:a2}
V_i'=\mathbf{C} (F_{G'}(i),Z).
\end{equation}

Furthermore, consider an arbitrary LU operator $U=U_1\otimes\cdots\otimes U_n$ that transforms $\ket{G}$ into $\ket{G'}$ and satisfies (\ref{eq:aux1}). Notice that 
\[
\widetilde{U}:=V'UV^\dagger
\]
stabilizer state $\ket{S}$ into stabilizer state $\ket{S'}$. Furthermore, we have
\begin{align*}
\widetilde{U_i}\;Z\;\widetilde{U_i}^\dagger &=
(V'_iU_iV^\dagger_i) \; Z\; (V_iU_i^\dagger {V_i'}^\dagger) 
\\
=&
\pm V'_i(U_i(V^\dagger_i \; Z\; V_i)U_i^\dagger ) {V_i'}^\dagger
\\
=&
\pm V'_i(U_i\; F_G(i)\; U_i^\dagger ) {V_i'}^\dagger
\\
=&
\pm V'_i\; F_{G'}(i)\; {V_i'}^\dagger
\\
=&
\pm Z  
\end{align*}
as $V_i, V_i'$ are transition matrices, see \cref{eq:a1,eq:a2}. As a consequence, $\widetilde{U_i}$ preserves Pauli-$Z$ matrix up to the phase. As we have shown in \cref{lemma:form_preserving_pauli} such matrix have a specific form, namely we have
\begin{equation}
\label{eq:ui}
\widetilde{U}_i =X^{s_i} Z^{\alpha_i},
\end{equation}
for some $s_i\in\{0,1\}$, continuous parameter $\alpha_i\in [0,2]$ and some phase $|\omega_i|=1$.

Notice that operator $\widetilde{U}$ is a product of two LU operators: $\widetilde{U}=\widetilde{U}^1\cdot\widetilde{U}^2$, where 
\begin{equation}
\label{eq:ui12}
\widetilde{U}_i^1 =\omega \prod_{i=1}^n Z^{\alpha_i}_i,
\quad
\widetilde{U}_i^2 =\prod_{i=1}^n X^{s_i}_i, 
\end{equation}
where $\omega:=\prod_{i=1}^n\omega_i$ is a global phase. 
We define the following stabilizer state:
\begin{equation}
\label{eq:stabBis}
\ket{\St''}:=\widetilde{U}_i^2 \ket{\St'}
\end{equation}
see \cref{fig3} for schematic diagram. 
Notice, that $\ket{\St}:=\widetilde{U}_i^1 \ket{\St''}$, hence two stabilizer states $\ket{\St}$ and $\ket{\St''}$ are connected by a local phase operator $\widetilde{U}_i^1$ \footnote{Notice, that in that case, we constructed a pair of stabilizer states that are connected by a local phase operator (\ref{eq:ui12}). This was already observed in Ref.~\cite[Theorem 1]{10.5555/2011763.2011766}, where the authors showed that any counterexample to the ``LU=LC conjecture'' is LC equivalent to the pair of stabilizer states connected by a local diagonal matrices.}. 
In accordance to \cref{lemma:Auxi_0_pm}, stabilizer $\St$ related to a stabilizer state $\ket{\St}$ is generated by some generators $\Sg=\langle s_1,\ldots,s_n\rangle$  with the corresponding check-matrix $[\X_{\Sg}| \Z_{\Sg}]$ in a normal form, i.e.: 
\begin{equation}
\label{eq:sghat_main}
\X_{\Sg}=
\begin{bNiceArray}{c|c}
\Block{1-1}{\id_k}&\Gamma_{k|n-k}^{}\\
\cline{1-2}
\Block{1-1}{0}&\Block{1-1}{0}\\
\end{bNiceArray},
\quad
\Z_{\Sg}=
\begin{bNiceArray}{c|c}
\Block{1-1}{\Gamma_{k|k}}&\Block{1-1}{0}\\
\cline{1-2}
\Block{1-1}{\Gamma_{k|n-k}^{\text{T}}}&\Block{1-1}{\id_{n-k}}\\
\end{bNiceArray},
\end{equation}
for some symmetric matrix $\Gamma_{k|k}^{}$ and arbitrary $k\times n-k$ matrix $\Gamma_{k|n-k}^{}$. 
Furthermore, as states $\ket{\St''}$ and $\ket{\St}$ are related by the phase operator (\ref{eq:ui12}), by \cref{lemma:Auxi_3.5_pm}, $\Sg''$ is generated by $\Sg''=\langle s_1'',\ldots,s_n''\rangle$ with the corresponding check-matrix $[\X_{\Sg}| \Z_{\Sg}]$ in the following form:
\begin{equation}
\label{eq:sghat_main2}
\X_{\Sg''}=
\begin{bNiceArray}{c|c}
\Block{1-1}{\id_k}&\Gamma_{k|n-k}^{}\\
\cline{1-2}
\Block{1-1}{0}&\Block{1-1}{0}\\
\end{bNiceArray},
\quad
\Z_{\Sg''}=
\begin{bNiceArray}{c|c}
\Block{1-1}{\Gamma_{k|k}'}&\Block{1-1}{0}\\
\cline{1-2}
\Block{1-1}{\Gamma_{k|n-k}^{\text{T}}}&\Block{1-1}{\id_{n-k}}\\
\end{bNiceArray},
\end{equation}
for some symmetric matrix $\Gamma_{k|k}'$. 

Notice that states $\ket{\St''}$ and $\ket{\St'}$ are related by operator (\ref{eq:stabBis}), hence the following set of stabilizers $\Sg'=\langle s_1',\ldots,s_n'\rangle$, $s_i':= \widetilde{U}_i^2 s_i''( \widetilde{U}_i^2)^\dagger$ generates stabilizer state $\ket{\St'}$. Notice that $\widetilde{U}_i^2$, see \cref{eq:ui12}, is a Pauli string consisting of $\id$ and $X$ operators only, hence $s_i'=\pm s_i''$ for all $i\in [n]$. Therefore, the generators $\Sg'=\langle s_1',\ldots,s_n'\rangle$ have the same check matrix as generators $\Sg''$, i.e. 
\begin{equation}
\label{eq:333}
[\X_{\Sg'}| \Z_{\Sg'}]=[\X_{\Sg''}| \Z_{\Sg''}]'
\end{equation}

Notice that stabilizers generators $\Sg=\langle s_1,\ldots,s_n\rangle$ are Pauli strings with real phases, i.e. $\phi (s_i)=\pm 1$. 
Indeed, stabilizer of any graph state is generated from standard generators $ g_1,\ldots,g_n$ with trivial phases, and any product of such generators can have the $\pm 1$ phase only. Therefore, the stabilizer of a graph state $\ket{G}$ contains Pauli strings with phases $\pm 1$ only. 
Similarly, stabilizer $\St$ of the stabilizer state $\ket{\St}=V \ket{G}$ is generated by $ V^\dagger g_1 V ,\ldots,V^\dagger g_n V$. By \cref{lemma:XXX}, those are Pauli strings with phases $\pm 1$ only, hence any element of stabilizer $\St$ has a phase $\pm 1$. In particular, generators $\Sg=\langle s_1,\ldots,s_n\rangle$ are Pauli strings with phases $\pm 1$, and similarly, generators $\Sg'=\langle s_1',\ldots,s_n'\rangle$. 

So far, we constructed stabilizer states $\ket{\St}$, and $\ket{\St'}$ connected by operator (\ref{eq:ui12}), and generated by generators $\Sg=\langle s_1,\ldots,s_n\rangle$ and $\Sg'=\langle s_1',\ldots,s_n'\rangle$ respectively, with check matrices in a normal form: \cref{eq:sghat_main,eq:sghat_main2,eq:333}. Furthermore, we observed that Pauli strings $s_i,s_i'$ have always phases $\pm 1$ for $i\in [n]$. 
Therefore, by \cref{lemma:Auxi_2_pm}, there exists a local LC-operator $W=W_1\otimes\cdots\otimes W_n$ that transforms set $\Sg$ into $\hat{\Sg}=\langle \hat{s}_1,\ldots,\hat{s}_n\rangle$, $\hat{s}_i=W^\dagger s_i W$, for which $\hat{s}_i$ are Pauli strings with trivial phases, i.e. $\phi_{\hat{s}_i}=1$, and the associated check matrix is in a strong normal form. Furthermore, the check matrix for generators $\hat{\Sg}=\langle \hat{s}_1,\ldots,\hat{s}_n\rangle$ is identical to the check matrix of $\Sg=\langle s_1,\ldots,s_n\rangle$, see (\ref{eq:sghat_main}), with all non-zero diagonal entries in $\Gamma_{k|k}$ removed. Moreover, matrices $W_i$ are of the form 
\begin{equation}
\label{eq:aaa}
W_i= Z^{t_i} \mathbf{C} (Y,X)^{\ell_i},
\end{equation}
for some $t_i,\ell_i\in \{0,1\}$, and by the definition, we have $|\hat{\St}\rangle:=W\ket{\St}$. \cref{lemma:Auxi_2_pm} contains all the details how to compute the form of $W$ operator. 
Similarly, we can find an local LC-operator $W'$ that transforms $\ket{\St'}$ into $|\hat{\St'}\rangle:= W' |\St'\rangle$, where $\hat{\St'}$ is generated by $\hat{\Sg'}=\langle \hat{s}_1',\ldots,\hat{s}_n'\rangle$ is a generator set with all $\hat{s}_i'$ being Pauli strings with trivial phases, i.e. $\phi_{\hat{s}_i'}=1$, and the associated check matrix is in a strong normal form. 
Similarly to (\ref{eq:aaa}), matrix $W'=W_1'\otimes\cdots\otimes W_n'$ is of the form 
\begin{equation}
\label{eq:aaa2}
W_i'= Z^{t_i'} \mathbf{C} (Y,X)^{\ell_i'},
\end{equation}
for some $t_i',\ell_i'\in \{0,1\}$. 

As a consequence, two stabilizer states $|\hat{S}\rangle$ and  $|\hat{S'}\rangle$ have stabilizers generated by generators $\hat{\Sg}=\langle \hat{s}_1,\ldots,\hat{s}_n\rangle$ and $\hat{\Sg'}=\langle \hat{s}_1',\ldots,\hat{s}_n'\rangle$ in a strong normal form. Moreover operator
\[
\breve{U}:=W'\circ \widetilde{U} \circ W^{-1}
\]
transforms stabilizer state $|\hat{S}\rangle$ into stabilizer state $|\hat{S'}\rangle$. 
Furthermore, operator $\breve{U}$ is an LU operator, and combining \cref{eq:ui,eq:aaa,eq:aaa2}, we have $\breve{U}=\omega \breve{U}_1 \otimes\cdots\otimes \breve{U}_n$, where
\begin{equation}
\label{eq:ui2}
\breve{U}_i =
\big(Z^{t_i} {\mathbf{C} (Y,X)^{\ell_i}} \big)^\dagger
\,
X^{s_i} Z^{\alpha_i}
\,
(Z^{t_i'} \mathbf{C} (Y,X)^{\ell_i'})
,
\end{equation}
for some parameters $t_i,t_i',\ell_i,\ell_i', s_i\in\{0,1\}$
and continuous parameter $\alpha_i\in [0,2]$ and some phase $|\omega_i|=1$. 
It is a technical observation, see \cref{ob:Auxi_2.5_pm}, that regardless of the values $t_i,t_i',\ell_i,\ell_i'$, (\ref{eq:ui2}) is equal to 
$
\breve{U}_i =\pm X^{s_i'} Z^{\alpha_i'}
$, 
for some $s_i'\in\{0,1\}$
and continuous parameter $\alpha_i'\in [0,2]$. 
Notice that the sight $\pm 1$ might be incorporated in a global phase of $\breve{U}$, hence we have that $\breve{U}=\omega' \breve{U}_1 \otimes\cdots\otimes \breve{U}_n$, where
\begin{equation}
\label{eq:ui3}
\breve{U}_i = X^{s_i'} Z^{\alpha_i'},
\end{equation}
for some $s_i'\in\{0,1\}$, and continuous parameter $\alpha_i'\in [0,2]$, and phase $|\omega'|=1$. 

Notice that (\ref{eq:ui3}) provides an LU transformation between two stabilizer states $|\hat{\St}\rangle$ and $|\hat{\St'}\rangle$ with the corresponding stabilizer generators $\hat{\Sg}$ and $\hat{\Sg'}$ has trivial phase vectors and check matrices in a strong normal form. 
All the transformations are schematically depicted on \cref{fig3}. 
Generators $\hat{\Sg'}$ and $\hat{\Sg}$ correspond to the check-matrices (\ref{eq:sghat_main}) and (\ref{eq:sghat_main2}) with removed off-diagonal entries in matrices $\Gamma_{k|k}^{},\Gamma_{k|k}'$. 
In accordance to \cref{lemma:Auxi_4_pm}, by applying LC-operator
\[
\hat{V}= \id^{\otimes k} \otimes H^{\otimes n-k},
\]
to both stabilizer states $|\hat{\St}\rangle$ and $|\hat{\St'}\rangle$, we obtain graph states $|\hat{G}\rangle:= \hat{V}|\hat{\St}\rangle$ and $\\\hat{G'}\rangle:= \hat{V}|\hat{\St'}\rangle$ corresponding to graphs with the adjacency matrices presented in (\ref{auxiii_main}). 
Simple computation shows that for arbitrary $s=\{0,1\}$ and $\alpha\in [0,2]$, we have $H^\dagger X^s Z^\alpha H=Z^s X^\alpha $, 
notice that coefficients $s,\alpha$ are the same on the left-hand, and right-hand side. 
This shows that graph states $|\hat{G}\rangle$ and $|\hat{G'}\rangle$ are LU equivalent by the operator
$\hat{U}'=\omega' \hat{U}_1' \otimes\cdots\otimes \hat{U}_n'$, where 
\begin{equation}
\label{eq:ui4}
\hat{U}_i' = 
\begin{cases}
  X^{s_i'} Z^{\alpha_i'}  & \text{ for }i\leq k  \\
  Z^{s_i'} X^{\alpha_i'}  & \text{ for }i> k
\end{cases}
\end{equation}
for the same coefficients $s_i'\in\{0,1\}$, continuous parameters $\alpha_i'\in [0,2]$, and phase $|\omega'|=1$ as in (\ref{eq:ui3}). 

Using\cref{prop:solution_1}, we conclude that graph states $|\hat{G}\rangle$ and $|\hat{G'}\rangle$ are LU equivalent by the operator
$\hat{U}=\omega\, \hat{U}_1 \otimes\cdots\otimes \hat{U}_n$, where 
\begin{equation}
\hat{U}_i = 
\begin{cases}
  Z^{\alpha_i}  & \text{ for }i\leq k  \\
   X^{\alpha_i}  & \text{ for }i> k
\end{cases}
\end{equation}
for the same parameters $\alpha_i\in [0,2]$, and phase $|\omega|=1$. 
Notice that $\omega$ is a global phase of operator $\hat{U}$ and does not affect quantum states, and hance can be omitted. 
This concludes the statement of \cref{prop:reduction_from_given_MLS}. 
\end{proof}

\begin{figure}
\centering
\includegraphics[width=0.45\textwidth]{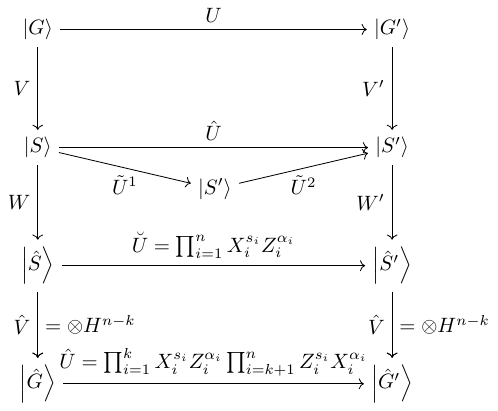}
\caption{Schematic description of all steps described in the proof of \cref{prop:reduction_from_given_MLS}. Suppose that graph states $\ket{G},\ket{G'}$ are LU equivalent, i.e. $U\ket{G}=\ket{G'}$ by operator $U$ that satisfies (\ref{eq:aux1}). Firstly, we construct stabilizer states $\ket{\St},\ket{\St'}$ that are LU equivalent by the matrix $\widetilde{U}$ of the form (\ref{eq:ui}). Secondly, we construct stabilizer state $\ket{\St''}$, which is LU equivalent to $ket{\St'}$ by a phase matrix $\widetilde{U}^2$ a pair of stabilizer states $\ket{\St''}$ and $\ket{\St'}$ that are LU equivalent by a phase matrix $\widetilde{U}^2$. Thirdly, we observe that stabilizers related to $\ket{\St''}$ and $\ket{\St'}$, and by extension stabilizer related to $\ket{\St}$, can be generated by generators $\Sg,\Sg'$, and $\Sg'' $ respectively, all with the same $\X$-part of check matrix. Fourthly, we transform stabilizer states $\ket{\St},\ket{\St'}$ to $|\hat{\St}\rangle,|\hat{\St'}\rangle$ with related stabilizers generated by generators $\hat{\Sg}$ and $\hat{\Sg'}$ both in a strong normal form. We observe that $|\hat{S}\rangle,|\hat{S'}\rangle$ are related by the matrix $\breve{U} $ which has the same structure as $\widetilde{U}$. Finally, we transform states $|\hat{S}\rangle,|\hat{S'}\rangle$ to graph states $|\hat{G}\rangle,|\hat{G'}\rangle$ which are equivalent by operator $\hat{U}'$ of the form (\ref{eq:ui4}). Lastly, using \cref{prop:solution_1}, we observe that graph states $|\hat{G}\rangle,|\hat{G'}\rangle$ are LU equivalent by operator $\hat{U}$ of the form (\ref{eq:ZX_new}), which finishes the proof of \cref{prop:reduction_from_given_MLS}.
}
\label{fig3}
\end{figure}

As we shall see, \cref{corollary:all_are_that_form} follows easily form \cref{prop_pairs_of_func} and \cref{prop:reduction_from_given_MLS}

\begin{proof}[Proof of \cref{corollary:all_are_that_form}]
Consider two graphs $G,G'$ such that the corresponding graph states are LU-equivalent by some operator $U=U_1\otimes\cdots\otimes U_n$. Consider arbitrary MLS cover $M$ for $G$. Notice that $M$ is also MLS cover for $G'$, by applying \cref{prop_pairs_of_func}, we identify pairs of functions $F_G^{\ell},F_{G'}^{\ell}:[n]\rightarrow \{X,Y,Z\}$ indexed by parameter $\ell\in \mathcal{L}$, such 
    \begin{align}
\label{eq:fG_fG'_functions_only_if_bis}
F_{G'}^\ell (i) = U _i \,F_{G}^\ell (i) \,U_i^\dagger  
\end{align}
for at least one $\ell\in \mathcal{L}$. 
We apply \cref{prop:reduction_from_given_MLS} for $G,G'$ and $F_{G}^\ell, F_{G'}^\ell$, which proves the statement.
\end{proof}

Lastly, we show that \cref{Th:reduction} follows from \cref{prop:reduction_from_given_MLS}. 

\begin{proof}[Proof of \cref{Th:reduction}] 
Consider two graph states $\ket{G}$ and $\ket{G'}$, which are LU-equivalent by operator $U=U_1\otimes\cdots\otimes U_n$ that satisfies 
\[
F_{G'} (i) = U _i \,F(i) \,U_i^\dagger
\]
for given functions $F_G,F_{G'}:[n]\rightarrow \{X,Y,Z\}$, and for all $i\in [n]$. 
\cref{prop:reduction_from_given_MLS} presents a step by step construction of graph $\hat{G}$ and $\hat{G'}$ that satisfies the statement of \cref{Th:reduction}. We presented those steps in a more accessible way in Algorithm 2. Notice that all steps in Algorithm 2 are simple row and column operations on adjecency matrices, together with Gaussian elimination. Therefore, the computational cost of performing Algorithm 2 is $\mathcal{O}(n^3)$.  
\end{proof}

\section{Transformation of (\ref{eq:11}) into a linear system over $\mathbb{F}_2$}
\label{Appendix_X}
In this section, we will discuss how to reformulate the linear system of equations (\ref{eq:11}) as a linear system over $\mathbb{F}_2$. Consider system (\ref{eq:11}) with solutions in $\alpha_i\in V_r$. We multiply equations in (\ref{eq:11}) by factor $2^{r-1}$ and introduce new variables $\hat{\alpha}_i := 2^{r - 1}\alpha_i$. In that way, we obtain a system of equations with $n-k$ integer variables $ \hat{\alpha}_i \in [2^{\Delta_G}]$ in modular arithmetic (mod $2^{r}$):
\begin{equation}
\label{eq:33}
    A\cdot\Vec{\hat{\alpha}} =\Vec{\hat{b}} \quad (\text{mod}\, 2^{r} ),
\end{equation}
where $A$ is the same integer matrix as in (\ref{eq:11}), and $\Vec{\hat{b}}: =2^{r-1}  \Vec{b}$ is an integer vector. By \cref{lemma:reduction_LC_r}, there are at most $(n-k)\sum_{t=2}^r\binom{\Delta_G}{t}$ non-trivial equations and $n - k$ integer variables in modular arithmetic modulo $2^{r}$. This system can then be transformed into a linear system over $\mathbb{F}_2$ with $r(n-k)\sum_{t=2}^r\binom{\Delta_G}{t}$ equations with $r(n-k)$ variables\footnote{This transformation is achieved by expressing all variables and coefficients in (\ref{eq:33}) in binary form. For instance, a variable $\alpha_i$ can be uniquely represented as $\alpha_i = \alpha{i0} 2^0 + \cdots + \alpha_{i, r-1} 2^{r-1}$, where $\alpha_{ij}$ are binary coefficients. Subsequently, each equation modulo $2^{r}$ can be rewritten as $r$ equations in $\mathbb{F}_2$ arithmetic}. 
This results in the following.

\begin{corollary}
\label{corollary:reduction_Further3}
Fix parameter $r\in[n]$ and consider the linear system of equations (\ref{eq:11}) with solutions restricted to $\alpha_i\in V_{r}$. Such system can be reformulated as a linear system over $\mathbb{F}_2$, with $r(n-k)$ binary variables and $r(n-k)\sum_{t=2}^r\binom{\Delta_G}{t}$ equations. 

Moreover, the number of variables can be upper bounded by $n^2$, while the number of equations by $n^{r+2}$.
\end{corollary}

We conclude this section with an additional consequence of \cref{proposition:reduction_Further2}, which, while not directly related to verifying LU-equivalence between graph states, reveals that certain graph states are stabilized by a continuous family of LU operators.

\begin{corollary}
\label{corollary:automorphism}
Consider a graph $\hat{G}$ with adjacency matrices of the form (\ref{auxiii_main_main}). If two vertices $i, j \in [k,n]$ have the same neighborhood, i.e., $\delta_i = \delta_j$, then the operators $X_i^\alpha X_j^{-\alpha}$ stabilize the corresponding graph state, i.e., $X_i^\alpha X_j^{-\alpha} \ket{\hat{G}} = \ket{\hat{G}}$ for any parameter $\alpha$.
\end{corollary}


\bibliography{Physics.bib}

\newpage
\begin{figure*}
    \centering
    \includegraphics[width=0.99\linewidth]{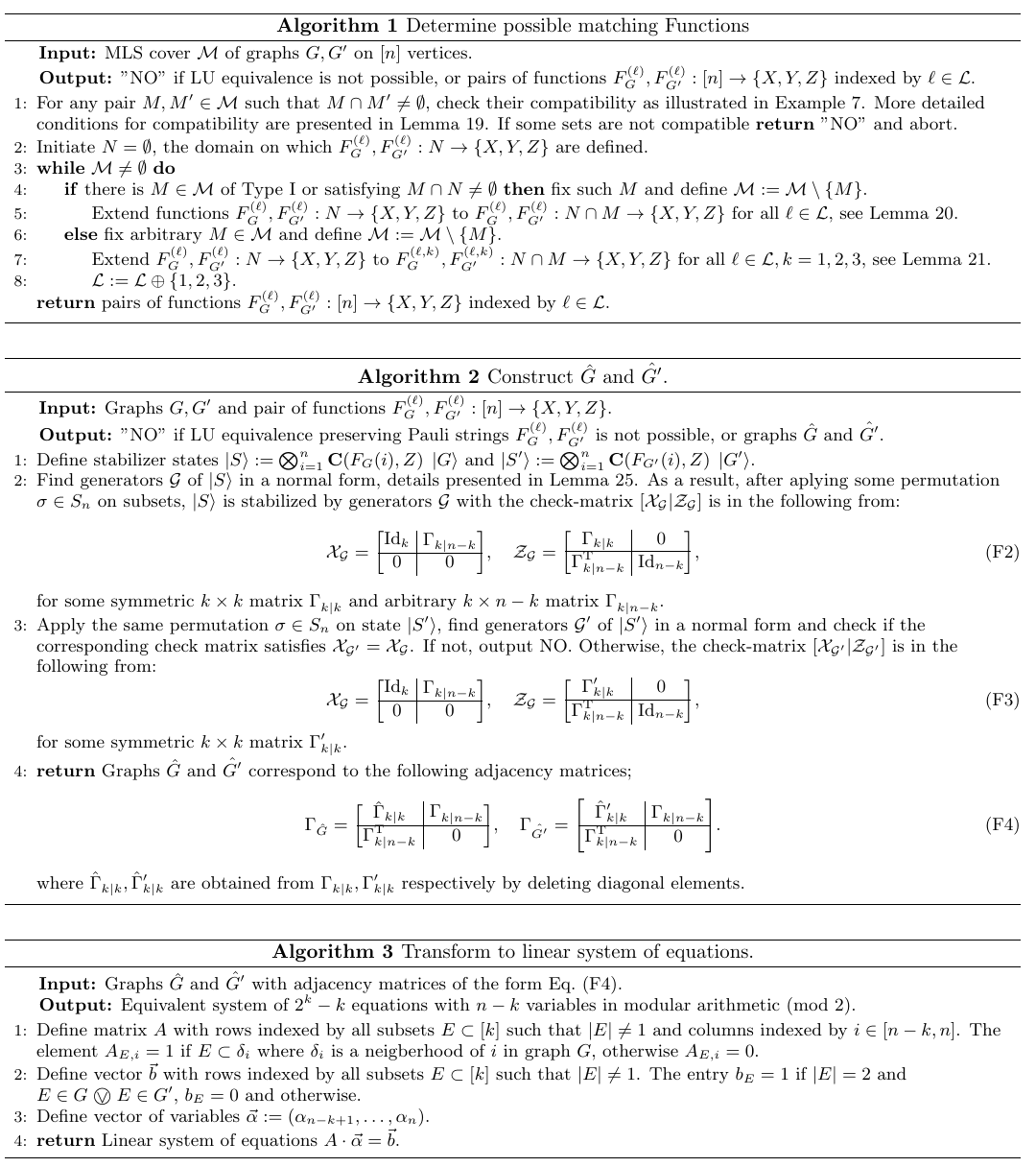}
    \caption{Pseudocode for three subroutines used in Algorithm 4 to verify LU equivalence between graph states $G$ and $G'$. They are related to three reduction steps graphica presented in \cref{fig:algo_outline}. 
   }
    \label{fig:algo_subroutines}
\end{figure*}

\begin{figure*}
    \centering
    \includegraphics[width=0.99\linewidth]{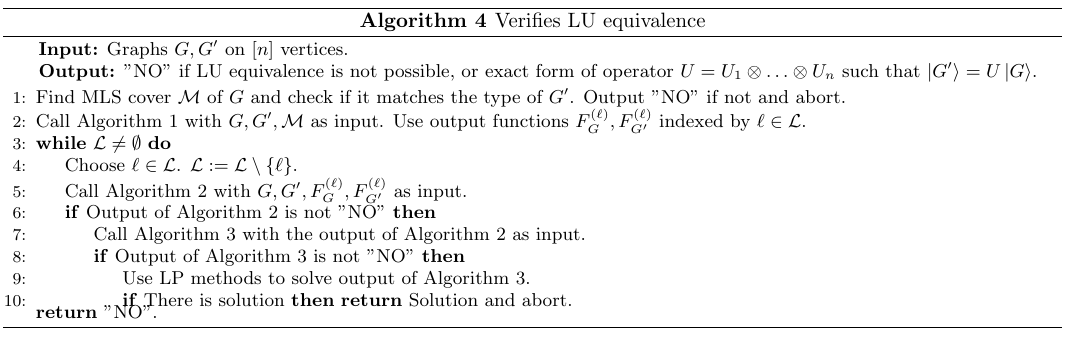}
    \caption{Pseudocode for an algorithm to verify LU equivalence between graph states $G$ and $G'$. Note that the algorithm for finding the MLS cover, along with the three algorithms presented in \cref{fig:algo_subroutines} are used as subroutines.}
    \label{fig:algo_main}
\end{figure*}

\end{document}